\documentclass[11pt,letterpaper]{article}

\newcommand{\suppress}[1]{}
\newcommand{\set}[1]{\left\{ #1 \right\}}
\usepackage{multicol}
\usepackage{listings}

\usepackage{fullpage,authblk,amsthm}
\newtheorem{theorem}{Theorem}[section]
\newtheorem{definition}[theorem]{Definition} 
\newtheorem{lemma}[theorem]{Lemma}
\newtheorem{proposition}[theorem]{Proposition}
\newtheorem{corollary}[theorem]{Corollary}

\usepackage{amsfonts,amsmath,amssymb,boxedminipage}

\usepackage[latin1]{inputenc}
\usepackage{url}

\usepackage{graphicx} 
\usepackage{algorithm}
\usepackage{algpseudocode}

\providecommand{\Exp}{\mathbb{E}}

\providecommand{\zo}{\{0, 1\}}
\renewcommand{\Pr}{\mathop{\mathrm{Pr}}}

\newcommand{\TMAJ}{\mathsf{3\textrm{-}MAJ}}
\newcommand{\MAJ}{\mathsf{MAJ}}
\newcommand{\NAND}{\mathsf{NAND}}

\newcommand{\algoone}{{Algorithm~\ref{alg:algoone}}}
\newcommand{\algotwo}{{Algorithm~\ref{alg:algotwo}}}

\newcommand{\Input}[1]{\State \textbf{Input:} #1}
\newcommand{\Output}[1]{\State \textbf{Output:} #1}
\newcommand{\BlankLine}{\State \ }

\newcommand{\evaluate}{\textsc{Evaluate}}
\newcommand{\complete}{\textsc{Complete}}
\newcommand{\m}{\mathrm{m}}
\newcommand{\rM}{\mathrm{M}}
\newcommand{\rR}{\mathrm{R}}
\newcommand{\rC}{\mathrm{C}}
\newcommand{\rD}{\mathrm{D}}
\newcommand{\rT}{\mathrm{T}}
\newcommand{\cH}{\mathcal{H}}

\newcommand{\etal}{\textit{et al.}}

\newcommand{\Order}{\mathrm{O}}

\begin{document}

\title{Improved bounds for the randomized decision tree complexity\\
of recursive majority%
~\thanks{This work presents an extension of the ideas reported
in~\cite{mnsx11}.
Partially supported by the French ANR Blanc project ANR-12-BS02-005 (RDAM)
and the European Commission IST STREP projects Quantum Computer Science (QCS) 255961
and Quantum Algorithms (QALGO) 600700.}}
\date{}

\author[1]{Fr\'ed\'eric Magniez}
\author[2]{Ashwin Nayak\thanks{%
Work done in part at Perimeter Institute for Theoretical Physics, Waterloo,
ON, Canada; LRI---CNRS, Universit{\'e} Paris-Sud, Orsay, France; and Centre for
Quantum Technologies, National University of Singapore, Singapore.
Partially supported by NSERC Canada. Research at PI is supported 
by the Government of Canada through Industry Canada and by the Province 
of Ontario through MRI.
}}
\author[1,3]{Miklos Santha\thanks{%
Research at the Centre
for Quantum Technologies is funded by the Singapore Ministry of Education 
and the National Research Foundation, also through the Tier 3 Grant ``Random numbers from quantum processes".
}}
\author[4]{\\ Jonah Sherman}
\author[5]{G\'abor Tardos\thanks{Research partially supported by the
  MTA RAMKI Lend{\"u}let Cryptography Research Group, NSERC, the Hungarian
  OTKA grant NN-102029 and an
  exchange program at Zheijang Normal University.}}
\author[1]{David Xiao}
\affil[1]{CNRS, LIAFA, Univ Paris Diderot,
Paris, France
}
\affil[2]{C\&O and IQC, University of Waterloo, Waterloo, Canada
}
\affil[3]{Centre for Quantum Technologies,
National U. of Singapore, Singapore 
}
\affil[4]{CS Division, University of California, Berkeley, USA}
\affil[5]{R\'enyi Institute, Budapest, Hungary}

\maketitle

\begin{abstract}
We consider the randomized decision tree complexity of the recursive
3-majority function.  We prove a lower bound of~$(1/2-\delta) \cdot
2.57143^h$ for the two-sided-error randomized decision tree complexity
of evaluating height $h$ formulae with error~$\delta \in [0,1/2)$.
This improves the lower bound of $(1-2\delta)(7/3)^h$
given by Jayram, Kumar, and Sivakumar (STOC'03),
and the one of $(1-2\delta) \cdot 2.55^h$ given by Leonardos (ICALP'13).
Second, we improve the upper bound by giving a new zero-error
randomized decision tree algorithm that has complexity at most
$(1.007) \cdot 2.64944^h$. The previous best known
algorithm achieved complexity $(1.004) \cdot 2.65622^h$.
The new lower bound follows from a better analysis of the base case of the
recursion of Jayram \etal{} The new algorithm uses a novel
``interleaving'' of two recursive algorithms.
\end{abstract}

\section{Introduction}
\label{sec-introduction}

Decision trees form a simple model for computing boolean functions by
successively reading the input bits until the value of the function
can be determined. In this model, the only cost we consider the number of input
bits queried. This allows us to study the complexity of computing a
function in terms of its structural properties.
Formally, a {\em deterministic decision tree algorithm\/}
$A$ on $n$ variables is a binary tree in which each internal node is
labeled with an input variable $x_i$, and the leaves of the tree are
labeled by either 0 or 1.  Each internal node has two outgoing edges,
one labeled with 0, the other with 1. Every input $x= x_1 \dots x_n$
determines a unique path in the tree leading from the root to a leaf:
if an internal node is labeled by~$x_i$, we follow either the~$0$
or the~$1$ outgoing edge according to the value of $x_i$. The value of the
algorithm $A$ on input $x$, denoted by $A(x)$, is the label of the
leaf on this unique path. Thus, the algorithm $A$ {\em computes} a boolean
function $A : \{0,1\}^n \rightarrow \{0,1\}$.

We define the {\em cost} $\rC(A,x)$ of a deterministic decision tree 
algorithm $A$ on input $x$ as the number of input bits queried by $A$ 
on $x$. Let ${\cal P}_f$ be the set of all deterministic decision tree 
algorithms which compute $f$.  The {\em deterministic complexity} of $f$ is 
$
\rD(f) = \min_{A \in {\cal P}_f} \max_{x \in \{0,1\}^n} \rC(A,x)
$.
Since every function can be evaluated after reading all the input variables, $\rD(f) \leq n$.

In an extension of the deterministic model, we
can also permit randomization in the computation. 
A {\em randomized decision tree algorithm} $A$ on $n$ variables is a
distribution over all deterministic decision tree algorithms on $n$
variables. Given an input $x$, the algorithm first samples a
deterministic tree $B \in_\rR A$, then evaluates $B(x)$.  The error
probability of $A$ in computing $f$ is given by $\max_{x \in \zo^n}
\Pr_{B \in_\rR A}[B(x) \neq f(x)]$.  The {\em cost} of a randomized
algorithm $A$ on input $x$, denoted also by $\rC(A,x)$, is the
expected number of input bits queried by $A$ on $x$.
Let ${\cal
  P}^\delta_f$ be the set of randomized decision tree algorithms
computing $f$ with error at most $\delta$. The two-sided
bounded error {\em randomized complexity\/} of $f$ with error~$\delta \in
[0,1/2)$ is
$
\rR_\delta(f) = \min_{A \in {\cal P}^\delta_f} \max_{x \in \{0,1\}^n} \rC(A,x)
$.

We write $\rR(f)$ for~$\rR_0(f)$.  By definition, for all $0 \leq \delta
< 1/2$, it holds that $\rR_\delta(f) \leq \rR(f) \leq \rD(f)$, and
it is also known \cite{BI87, HH87, T90} that $\rD(f) \leq \rR(f)^2$,
and that for all constant $\delta \in (0,1/2)$, $\rD(f) \in
\Order(\rR_\delta(f)^3)$ \cite{Nis89}.

Considerable attention in the literature has been given to the
randomized complexity of functions computable by read-once formulae,
which are boolean formulae in which every input variable appears only
once.  For a large class of well balanced formulae with $\NAND$ gates
the exact randomized complexity is known. In particular, let $\NAND_h$
denote the \emph{complete} binary tree of height $h$ with $\NAND$
gates, where the inputs are at the $n = 2^h$ leaves.  Snir
\cite{Snir85} has shown that $\rR(\NAND_h) \in \Order(n^c)$ where $c =
\log_2 \Bigl(\frac{1+ \sqrt{33}}{4}\Bigr) \approx 0.753$.
A matching $\Omega(n^c)$
lower bound was obtained by Saks and Wigderson~\cite{SaksW86}, and
extended to Monte Carlo algorithms (i.e., with constant error $\delta<1/2$) by
Santha~\cite{Santha95}.
Since $\rD(\NAND_h) = 2^h = n$ this implies that $\rR(\NAND_h) \in \Theta (
\rD(\NAND_h)^c)$. Saks and Wigderson conjectured that
this is the largest gap between deterministic and randomized
complexity:
\emph{for every} boolean function $f$ and constant~$\delta \in [0,1/2)$,
$\rR_\delta(f) \in \Omega ( \rD(f)^c)$.
For the zero-error (Las Vegas) randomized complexity of read-once threshold
formula of depth $d$, Heiman, Newman, and Wigderson~\cite{HeNW90}
proved a lower bound of $\Omega (n/2^d)$. Heiman and
Wigderson~\cite{HeW91} proved that the zero-error randomized complexity of every
read-once formula $f$ is at least $\Omega ( \rD(f)^{0.51}).$

After such progress,
one would have hoped that the simple model of decision tree algorithms 
might shed more light on the power of randomness. But surprisingly, we 
know the exact randomized complexity of
very few boolean functions. In particular, the randomized complexity of
the recursive 3-majority function ($\TMAJ_h$) is still open. This function,
proposed by Boppana, was one of the earliest examples where randomized
algorithms were found to be more powerful than deterministic decision
trees~\cite{SaksW86}. It is a read-once formula on $3^h$ variables given 
by the complete ternary tree of height $h$ whose internal nodes are
majority gates.
It is easy to check that $\rD(\TMAJ_h) = 3^h$, but there is a naive
randomized recursive algorithm for $\TMAJ_h$ that performs better:
pick two random children of the root and recursively evaluate them,
then evaluate the third child if the value is not yet determined.
This has zero-error randomized complexity $(8/3)^h$. However, it was
already observed by Saks and Wigderson~\cite{SaksW86} that one can do
even better than this naive algorithm.  As for lower bounds,
that reading $2^h$ variables is necessary for zero-error algorithms is
easy to show. In spite of some similarities with the $\NAND_h$
function, no progress was reported on the randomized complexity of
$\TMAJ_h$ for~17 years.
In 2003, Jayram, Kumar, and Sivakumar~\cite{JaKS03} proposed an 
explicit randomized algorithm that achieves complexity $(1.004) \cdot 2.65622^h$,
and beats the naive recursion.
(Note, however, that the recurrence they derive in \cite[Appendix~B]{JaKS03}
is incorrect.)
They also prove a $(1-2\delta)(7/3)^h$
lower bound for the $\delta$-error randomized decision tree complexity 
of $\TMAJ_h$. In doing so, they introduce a powerful combinatorial
technique for proving decision tree lower bounds.

In this paper, we considerably improve the lower bound obtained in
\cite{JaKS03}, first by proving that
$\rR_\delta(\TMAJ_h) \geq (1-2\delta)(5/2)^h$, then further improving
the base 5/2. 
We also improve the upper bound by giving a new zero-error randomized decision
tree algorithm.
\begin{theorem}
For all~$\delta \in [0,1/2]$, we have
\[
(1/2-\delta)\cdot 2.57143^h
    \quad \leq \quad \rR_\delta(\TMAJ_h) 
    \quad \leq \quad (1.007) \cdot 2.64944^h \enspace.
\]
\end{theorem}
In contrast to the randomized case, the bounded-error \emph{quantum\/} 
query complexity of $\TMAJ_h$ is known more precisely; it is
in~$\Theta(2^h)$~\cite{ReichardtS08}.

\textbf{New lower bound.}
For the lower bound, Jayram
\etal{} consider a complexity measure related to the distributional
complexity of~$\TMAJ_h$ with respect to a specific ``hard'' distribution
(cf.\ Section~\ref{sec-JKS}).  The focus of the proof is a relationship
between the complexity of evaluating formulae of height~$h$ to that of
evaluating formulae of height~$h-1$. They derive a sophisticated
recurrence relation between these two quantities, that finally implies
that $\rR_\delta(\TMAJ_h) \geq\alpha (2 + q)^h$, where
$\alpha q^h$ is a lower bound on the probability~$p^\delta_h$ that
a randomized algorithm with error at most $\delta$ queries a special
variable, called the ``absolute minority'', on inputs drawn from the hard distribution.
They observe that any randomized decision tree with error at most $\delta$
queries at least one variable with probability~$1-2\delta$.  This
variable has probability~$3^{-h}$ of being the absolute minority,
so~$q=1/3$ and $\alpha=1-2\delta$ satisfies the conditions and their lower bound follows.

We obtain new lower bounds by improving the bound on~$p_h^\delta$.
We start by proving that~$p_h^\delta \geq (1-2\delta)2^{-h}$, i.e.,
increasing $q$ to $1/2$, which immediately implies a better lower 
bound for $\rR_\delta(\TMAJ_h)$.  To obtain this bound, we examine the 
relationship between~$p_h^\delta$ and~$p_{h-1}^\delta$, by encoding a 
height~$h-1$ instance into a height~$h$ instance, and using an algorithm 
for the latter instance.
Analyzing this encoding requires understanding the behavior of all
decision trees on $3$ variables, and this can be done by exhaustively
considering all such trees.

We further improve this lower bound by encoding height $h-2$ instances 
into height $h$ instances,
and prove $p^\delta_h\ge\alpha q^h$ for $q=\sqrt{7/24}>0.54006$. For
technical reasons we set $\alpha=1/2-\delta$ (half the value
considered by Jayram \etal{} in their bound).
For encodings of height $h-3$ and $h-4$ instances into 
height $h$ instances, we use a computer to get the better estimates,
with $q=(2203/12231)^{1/3} > 0.56474$ and $q=(216164/2027349)^{1/4}
> 0.57143$, respectively.

The lower bound of~$(1-2\delta)(5/2)^h$ mentioned above was presented in a
preliminary version of this article~\cite{mnsx11}.
Independent of the further improvements we make, Leonardos~\cite{leo13}
gave a lower bound of $\rR_\delta(\TMAJ_h)\geq (1-2\delta)\cdot 2.55^h$.
His approach is different from ours, and is based on the method 
of generalized costs proposed by Saks and Wigderson~\cite{SaksW86}.
The final lower bound~$(1/2-\delta)\cdot 2.57143^h$
we obtain surpasses the bound due to Leonardos.

\textbf{New algorithm.}
The naive algorithm and the algorithm of Jayram \etal{} are examples
of {\em depth-$k$\/} recursive algorithms for~$\TMAJ_h$, for~$k =
1,2$, respectively. A depth-$k$ recursive algorithm is a collection of
subroutines, where each subroutine evaluates a node (possibly using
information about other previously evaluated nodes), satisfying the
following constraint: when a subroutine evaluates a node~$v$, it is
only allowed to call other subroutines to evaluate children of $v$ at
depth at most $k$, but is not allowed to call subroutines or otherwise
evaluate children that are deeper than $k$. (Our notion of depth-one
is identical to the terminology ``directional'' that appears in the
literature. In particular, the naive recursive algorithm is a
directional algorithm.)

We present an improved depth-two recursive algorithm.
To evaluate the root of the
majority formula, we recursively evaluate one grandchild from each
of two distinct children of the root.  The grandchildren ``give an
opinion'' about the values of their parents.  The opinion guides the
remaining computation in a natural manner: if the opinion indicates
that the children are likely to agree, we evaluate the two children in
sequence to confirm the opinion, otherwise we evaluate the third
child.  If at any point the opinion of the nodes evaluated so far
changes, we modify  future computations accordingly. A key
innovation is the use of an algorithm optimized to compute the value
of a \emph{partially evaluated} formula.  In the analysis, we
recognize when incorrect opinions are formed, and take advantage of
the fact that this happens with smaller probability.

We do not believe that the algorithm we present here is optimal.
Indeed, we conjecture that even better algorithms exist that follow 
the same high level intuition applied to depth-$k$ recursion, for $k > 2$.
However, it seems new insights are required to analyze the performance 
of deeper recursions, as the formulas describing their
complexity become unmanageable for $k > 2$.

\textbf{Organization. }
We prepare the background for the main
results Section~\ref{sec-prelim}. In Section~\ref{sec-lb} we prove the new
lower bounds for $\TMAJ_h$. The new algorithm for the problem is described 
and analyzed in Section~\ref{sec-algo}. 

\section{Preliminaries}
\label{sec-prelim}

We write~$u \in_\rR D$ to state that~$u$ is sampled
from the distribution~$D$.  If $X$ is a finite set, we identify
$X$ with the uniform distribution over $X$, and so, for instance, $u
\in_\rR X$ denotes a uniform element of $X$.

\subsection{Distributional Complexity}

A variant of the randomized complexity we use is {\em
  distributional\/} complexity.  Let ${\cal D}_n$ be the set of
distributions over $\{0,1\}^n$. The {\em cost} $\rC(A,D)$ of a
randomized decision tree algorithm $A$ on $n$ variables with respect
to a distribution $D \in {\cal D}_n$ is the expected number of bits
queried by $A$, where the expectation is taken over inputs sampled
from $D$ and the random coins of $A$.  The \emph{distributional
complexity\/} of a function $f$ on $n$ variables for $\delta$ 
two-sided error is
$
\Delta_\delta(f) = \max_{D \in {\cal D}_n} \min_{A \in {\cal P}^\delta_f} \rC(A,D)
$.
The following observation is a well established route to proving 
lower bounds on worst case complexity.
\begin{proposition}
\label{fact:dist}
$\rR_\delta(f) \geq \Delta_\delta(f)$.
\end{proposition}
\suppress{
To see this fact, consider the distribution $D$ achieving the maximum
in the definition of $\Delta_\delta(f)$; then by an averaging
argument, for all $A \in {\cal P}^\delta_f$ there exists an input $x$
in the support of $D$ such that $C(A, x) \geq \rC(A, D) \geq
\Delta_\delta(f)$.  This implies that the worst-case complexity of all
$A$ is at least $\Delta_\delta(f)$.
}

\subsection{The $\TMAJ_h$ Function and the Hard Distribution}
\label{sec-hard-distr}

Let $\MAJ(x)$ denote the boolean majority function of its input bits.
The ternary majority function $\TMAJ_h$ is defined recursively on
$n=3^h$ variables, for every $h \geq 0$.  We omit the height $h$ when
it is obvious from context.  For $h=0$ it is the identity
function. For $h >0$, let $x$ be an input of length $n$ and let
$x^{(1)}, x^{(2)}, x^{(3)}$ be the first, second, and third $n/3$
variables of $x$. Then
$$
\TMAJ_h (x)
 \quad = \quad  \MAJ ( 
\TMAJ_{h-1}(x^{(1)}), \
 \TMAJ_{h-1}(x^{(2)}), \
\TMAJ_{h-1}(x^{(3)})
) \enspace.
$$
In other terms, $\TMAJ_h$
is defined by the read-once formula on the complete ternary tree
$\rT_h$ of height $h$ in which every internal node is a majority gate.
We identify the leaves of $\rT_h$ from left to right with the integers $1, \ldots, 3^h$. For an
input $x \in \{0,1\}^h$, the bit $x_i$ defines the {\em value} of the leaf  $i$, and then the
values of the internal nodes are evaluated recursively. The value of
the root is $\TMAJ_h(x)$. 
For every node $v$ in $\rT_h$ different from the root, let $P(v)$ denote
the parent of $v$.  We say that $v$ and $w$ are siblings if $P(v) =
P(w)$. For any node $v$ in $\rT_h$, 
let $Z(v)$ denote the set of variables associated with the leaves in
the subtree rooted at $v$.  We say that a node $v$ is at depth $d$ in
$\rT_h$ if the distance between $v$ and the root is $d$. The root is
therefore at depth 0, and the leaves are at depth $h$.

We now define recursively, for every $h \geq 0$, the set ${\cal H}_h$
of {\em hard inputs} of height $h$. 
In the base case ${\cal H}_0 = \{0,1\}$. For $h > 0$, let
\begin{equation*}
\begin{split}
{\cal H}_h \quad =  \quad \{(x,y,z) \in {\cal H}_{h-1} \times {\cal H}_{h-1}
\times {\cal H}_{h-1}  : &\  \TMAJ_{h-1}(x), \TMAJ_{h-1}(y),
\text{ and } \\
    &\quad \TMAJ_{h-1}(z) \text{ are not all identical}\} \enspace.
\end{split}
\end{equation*}
The hard inputs consist of instances for which at each
node $v$ in the ternary tree, one child of $v$ has value different from
the value of $v$. The {\em hard distribution\/} on inputs
of height $h$ is defined to be the uniform distribution over ${\cal
  H}_h$. We call a hard input $x$ {\em 0-hard} or {\em
  1-hard} depending on whether $\TMAJ_h(x)=0$ or $1$. We write
$\cH_h^0$ for the set of 0-hard inputs and $\cH_h^1$ for the set of
1-hard inputs.

For an $x \in {\cal H}_h$, the {\em minority path\/} $M(x)$ is the path,
starting at the root, obtained by following the child whose value
disagrees with its parent.  For $0\leq d \leq h$, the node of $M(x)$
at depth $d$ is called the depth $d$ minority node, and is
denoted by $M(x)_d$.  We call the leaf $M(x)_h$ of the minority path
the \emph{absolute minority} of $x$, and denote it by $m(x)$.

\subsection{The Jayram-Kumar-Sivakumar Lower Bound}
\label{sec-JKS}

For a deterministic decision tree algorithm $B$ computing $\TMAJ_h$, let
$L_B(x)$ denote the set of variables queried by $B$ on input $x$.
Recall that ${\cal P}^\delta_{\TMAJ_h}$ is the set of all randomized
decision tree algorithms that compute $\TMAJ_h$ with two-sided error at most $\delta$.
Jayram \etal{} define the function
$I^\delta(h,d)$, for $d \leq h$:
$$
I^\delta(h,d) \quad = \quad 
  \min_{A \in {\cal P}^\delta_{\TMAJ_h}} \Exp_{x \in_\rR
  \cH_h, B \in_\rR A} [| Z(M(x)_d) \cap L_B(x) |] \enspace.
$$
In words, it is the minimum over algorithms computing $\TMAJ_h$,
of the expected number of queries below the $d$th level minority node,
over inputs from the hard distribution. 
Note that $I^\delta(h,0) =  \min_{A \in {\cal P}^\delta_{\TMAJ_h}} \rC(A, \cH_h)$,
and therefore by Proposition~\ref{fact:dist},
$\rR_\delta(\TMAJ_h) \geq I^\delta(h,0)$.

We define $p^\delta_h = I^\delta(h,h)$,
which is the minimal probability that a $\delta$-error
algorithm $A$ queries the absolute minority of a random hard $x$ of
height $h$.

Jayram \etal{} prove a recursive lower bound for $I^\delta(h,d)$ using 
information theoretic arguments. A more elementary proof can be found in
~\cite{LNPV06}.
\begin{theorem}[Jayram, Kumar, Sivakumar~\cite{JaKS03}]
\label{thm: jks}
For all $0 \leq d < h$:
$$
I^\delta(h,d) \quad \geq \quad I^\delta(h,d+1) + 2 I^\delta(h-1, d)
\enspace.
$$
\end{theorem}
A simple computation gives then the following  lower bound on $I^\delta(h,d)$, for
all $0 \leq d \leq h$, expressed as a function of the $p^\delta_i$'s:
\begin{equation*}
I^\delta(h,d) \quad \geq \quad \sum_{i=d}^h {h-d \choose i-d} 2^{h-i}  p^\delta_i \enspace.
\end{equation*}
When $d=0$, this gives
$I^\delta(h,0) \geq \sum_{i=0}^h {h \choose i} 2^{h-i}  p^\delta_i$.
Putting this together with the fact that $\rR_\delta(\TMAJ_h) \geq
I^\delta(h,0)$, we get the following corollary: 
\begin{corollary}
\label{cor: jks}
Let $q, a > 0$ such that $p^\delta_i \geq a \cdot q^i$ for all~$i \in
\set{0, 1, 2, \dotsc, h}$. Then $\rR_\delta(\TMAJ_h) \geq a (2+q)^h$.
\end{corollary}
As mentioned in Section~\ref{sec-introduction}, Jayram \etal{} obtain
the lower bound of~$(1-2\delta)(7/3)^h$ from this corollary by observing that
$p^\delta_h \geq (1-2\delta) (1/3)^h$.

\section{Improved Lower Bounds}
\label{sec-lb}

\subsection{First Improvement}

In this section, we develop a method to enhance the
Jayram-Kumar-Sivakumar technique for establishing a lower bound
for~$\TMAJ$. The enhancement comes from an improved estimate
for~$p_h^\delta$, the minimum probability with which a decision tree
queries the absolute minority of an input drawn from the hard
distribution. 

\begin{theorem}
\label{thm:1level}
For every error $\delta > 0$ and height $h \geq 0$, we have 
$p^\delta_h \geq (1-2\delta) 2^{-h}$.
\end{theorem}
\begin{proof}
We prove this theorem by induction. Clearly, $p^\delta_0 \geq 1-2\delta$.
It then
suffices to show that $2 p^\delta_h \geq p^\delta_{h-1}$ for $h \geq 1$.
We do so by reduction as follows: let~$A$ be a randomized algorithm that
achieves the minimal probability $p^\delta_h$ for height~$h$ formulae.  We
construct a randomized algorithm $A'$ for height $h-1$ formulae such that the
probability that  $A'$ errs is at most $\delta$, and $A'$
queries the absolute minority with probability at most $2 p^\delta_h$.
Since $p^\delta_{h-1}$ is the minimum probability of querying the
absolute minority in the hard distribution, computed over all randomized
algorithms on inputs of height $h-1$ with error at most $\delta$,
this implies that $2 p^\delta_h \geq p^\delta_{h-1}$.

We now specify the reduction.  For the sake of simplicity, we omit
the error $\delta$ in the notation.  We use the following definition:
\begin{definition}[One level encoding scheme]
  A {\em one level encoding scheme\/} is a bijection $\psi: {\cal
    H}_{h-1} \times \{1, 2, 3\}^{3^{h-1}} \rightarrow {\cal H}_{h} $, such
  that for all $(y,r)$ in the domain, $\TMAJ_{h-1}(y) = \TMAJ_h(\psi(y,r))$.

  Let $c : \{0,1\} \times \{1,2,3\} \rightarrow {\cal H}_1$ be a
bijection
  satisfying $b = \MAJ(c(b,s))$ for all inputs $(b, s)$.  Define the
  one level encoding scheme $\psi$ {\em induced} by $c$ as follows:
  $\psi(y, r) = x \in \cH_h$ such that for all $1\leq i \leq 3^{h-1}$,
$
(x_{3i-2}, x_{3i-1}, x_{3i}) = c(y_i, r_i)
$.
\end{definition}
 
To define $A'$, we use the one level encoding scheme $\psi$ induced by
the following function: $c(y,1) = y01$, $c(y,2) = 1y0$, and $c(y,3) = 01y$.

On input $y$, algorithm $A'$ picks a uniformly random string $r \in
\{1, 2, 3\}^{3^{h-1}}$, and runs $A$ on $x =
\psi(y\mathrm{,} r)$ and computes the same output. Notice that each bit of $x_i$ of $x$ is either determined
by $r$ alone or else it is $y_{\lceil i/3\rceil}$. When $A$ asks for a bit of
$x$ that is determined by $r$, then this value is ``hard wired'' in $A'$ and
$A'$ makes no query. When $A$
asks for a bit of $x$ that is not determined by $r$, then and $A'$
queries the corresponding bit of $y$.
Observe that $A'$ has error at most $\delta$
as $\TMAJ_{h-1}(y) = \TMAJ_h(\psi(y, r))$ for all $r$, and $A$ has
error at most $\delta$.  We claim that
\begin{equation}
\label{eq:iosjfbn}
2 \Pr_{B  \in_\rR A,\; x \in_\rR \cH_h}[B(x)\text{ queries } x_{m(x)}] \quad \geq \quad
\Pr_{B  \in_\rR A,\; (y, r) \in_\rR \cH'_h}[B'(y, r)\text{ queries }y_{m(y)}]
\enspace,
\end{equation}
where~$\cH_h'$ is the uniform distribution 
over~$\mathcal{H}_{h-1} \times \{1,2,3\}^{3^{h-1}}$
and $B'$ is the algorithm that computes $x = \psi(y,r)$ and then evaluates $B(x)$.
We prove this inequality by taking an appropriate partition of the
probabilistic space of hard inputs ${\cal H}_{h}$, and prove
Eq.~(\ref{eq:iosjfbn}) separately, on each set in the partition.  For
$h=1$, the two classes of the partition are ${\cal H}_1^0$ and ${\cal
  H}_1^1$ .  For $h > 1$, the partition consists of the equivalence
classes of the relation $\sim$ defined by $x \sim x'$ if $x_i = x'_i$
for all $i$ satisfying $P(i) \neq P(m(x))$ in the tree $T$.

Because $\psi$ is a bijection, observe that this also induces a
partition of $(y, r)$, where $(y, r) \sim (y', r')$ if and only if $\psi(y, r)
\sim \psi(y', r')$.
Also observe that every equivalence class contains three elements.
Then Eq.~(\ref{eq:iosjfbn})
follows from the following stronger statement: for every equivalence class
$S$, and for all $B$ in the support of $A$, it holds that
\begin{equation}
\label{eq:naiweq}
2 \Pr_{x \in_\rR \cH_h } [B(x)\text{ queries } x_{m(x)} \mid x \in S ] 
\quad \geq \quad
\Pr_{(y,r) \in_\rR \cH_h'}[B'(y,r) \text{ queries }y_{m(y)} \mid \psi(y, r) \in S]
\enspace.
\end{equation}

The same proof applies to all sets $S$, but to simplify the notation,
we consider a set $S$ that satisfies the following: for $x \in S$, we
have $m(x) \in \{1,2,3\}$ and $x_{m(x)} = 1$.  Observe that for
each $j>3$, the $j$th bits of all three elements in $ S$ coincide.
Therefore, the restriction of $B$ to the variables $(x_1, x_2, x_3)$,
when looking only at the three inputs in $S$, is a well-defined
decision tree on three variables.  We call this restriction $C$, and formally it is
defined as follows: for each query $x_j$ made by $B$ for $ j >3$,
$C$ simply uses the value of $x_j$ that is
shared by all $x \in S$ and that we hard-wire into $C$; for each
query $x_j$ made by $B$ where $j \in 
\{1, 2, 3\}$, $C$ actually queries $x_j$.  
Note that the restriction $C$ does not necessarily compute
$\MAJ_1(x_1,x_2,x_3)$, for two reasons. Firstly, $C$ is derived
from $B$, which may err on particular inputs. But even if $B(x)$
correctly computes $\TMAJ_h(x)$, it might
happen that $B$ never queries any of $x_1, x_2, x_3$, or it might
query one and never query a second one, etc.

For any $x \in S$, recall that we write $(y, r)$ the unique solution of $\psi(y,r) = x$.  
It holds for our choice of $S$ that $m(y) = 1$ because we assumed $m(x)
\in \{1, 2, 3\}$ and also $y_1 = y_{m(y)} = 0$ because we assumed $x_{m(x)}
= 1$.

Observe that, for inputs $x \in S$, $B$ queries $x_{m(x)}$ if and only
if $C$
queries the minority among $x_1, x_2, x_3$.  Also, $B'(y, r)$ queries
$y_{m(y)}$ if and only if $C(\psi(0, r_1))$ queries $x_{r_1}$ (cf.\
definition of $c$).  Furthermore,
the distribution of $x_1 x_2 x_3$ when $x \in_\rR S$ is 
uniform over~$\mathcal{H}^0_1$.  Similarly, the
distribution of $r_1$ over uniform $(y, r)$ conditioned on $\psi(y, r)
\in S$ is identical to that of~$(0, r_1) = \psi^{-1}(x_1 x_2 x_3)$ for 
$x_1 x_2 x_3 \in_\rR \mathcal{H}^0_1$.
Thus Eq.~(\ref{eq:naiweq}) is equivalent to:
\begin{equation}
\label{eq:ijsfao} 
\Pr_{x \in_\rR {\cal H}_1^0}[C(x) \text{ queries }
x_{r_1}\text{ where } \psi(0, r_1) = x]  
    \quad \le \quad 2 \Pr_{x \in_\rR {\cal H}_1^0} 
                    [C(x)\text{ queries } x_{m(x)}]
    \enspace.
\end{equation}

In principle, one can prove this inequality by considering all the (finitely
many) decision trees $C$ on three variables. We present here a somewhat more
compact argument.

If $C$ does not query any bit, both sides of Eq.~(\ref{eq:ijsfao}) 
are zero, so the inequality holds.
We therefore assume that $C$ makes at least one query and, without loss of
generality, we also assume that the first query is $x_1$. We distinguish
two cases.
 
 If $C$ makes a second query when the first query is evaluated to $0$
 then the right hand side of Eq.~(\ref{eq:ijsfao}) is at least $4/3 = 2 \cdot
 (1/3 + 1/3)$ because
 there is a $1/3$ chance that the first query is $m(x)$ and $1/3$
 chance that the second is $m(x)$. But the left hand side is at most $1$,
 and therefore the inequality holds.

 If $C$ does not make a second query when the first query is evaluated
 to $0$ then the left hand side is at most $2/3$ since for $x=010$, we
 have $r_1 = 3$, but $x_3$ is not queried. With probability $1/3$ we have
 $m(x)=1$, so the right hand side is at least $2/3$. We conclude that
 Eq.~(\ref{eq:ijsfao}) holds for every $C$.

We remark that the decision tree algorithm making no queries is not the only
one that makes Eq.~(\ref{eq:ijsfao}) hold with equality. Another such algorithm
is the following:
first query $x_1$, if $x_1 = 0$, stop, else if $x_1 = 1$, query $x_2$ and stop.

 To handle a general $S$, we replace $\{1, 2, 3\}$ with $m(x)$ and its
 two siblings.  For $S$ such that $x \in S$ satisfies $x_{m(x)} = 0$, the
 optimal algorithm $C'$ is the same as the one described above, except
 that each $0$ is changed to $1$ and vice versa.

 Therefore Eq.~(\ref{eq:ijsfao}) holds for every $C$, which implies
 the theorem.

\end{proof}
Combining Corollary~\ref{cor: jks} and Theorem~\ref{thm:1level}, we obtain 
the following lower bound.
\begin{corollary}
\label{cor: main}
$\rR_\delta(\TMAJ_h) \geq (1-2\delta) (5/2)^h$.
\end{corollary}

\subsection{Further Improvement}
\label{sec:lb}

The proof of Theorem~\ref{thm:1level} proceeds by proving a recurrence,
using a one level encoding scheme, for the minimal probability~$p^\delta_h$ 
that an algorithm queries the absolute minority bit. It is natural to ask 
whether this is the best possible recurrence.  In this section, we show
that it is indeed possible to improve the recurrence by using 
higher level encoding schemes. 

In the following, we sometimes omit the error parameter~$\delta$ from 
the notation in the interest of readability. Let ${\cal R}= \{0,1\}
\times\{1,2,3\}$, ${\cal R}^{(1)}_h =
{\cal R}^{3^{h-1}}$, and~${\cal R}^{(k)}_h = {\cal R}^{(k-1)}_h
\times{\cal R}^{3^{h-k}}$.

\begin{definition}[Uniform $k$-level encoding scheme]
Let $c:\{0,1\}\times{\cal R}\to{\cal H}_1$ be the function given by $c(y,(b,1))
= yb(1-b)$, $c(y,(b,2)) = (1-b)yb$, and $c(y,(b,3)) = b(1-b)y$.
The {\em uniform $k$-level encoding scheme\/} $\psi^{(k)}$, for an
integer~$k \ge 1$ is defined by the following recursion:
\begin{enumerate}
\item For $h\ge1$, $y\in\cH_{h-1}$ and $r\in{\cal R}^{(1)}_h$ we
set $\psi^{(1)}(y,r) = x \in \cH_h$ such that $(x_{3i-2}, x_{3i-1}, x_{3i}) = c(y_i, r_{i})$, for all $1\leq i \leq 3^{h-1}$;
\item for $h\ge k>1$, $y\in\cH_{h-k}$ and $(R,r)\in{\cal R}^{(k)}_h$ we set
$\psi^{(k)}(y,(R,r))=\psi^{(k-1)}(\psi^{(1)}(y,r),R)$.
\end{enumerate}
\end{definition}

The uniform $2$-level encoding scheme is illustrated in
Figure~\ref{2-level}.
\begin{figure}[h]
\centerline{\includegraphics[height=5cm]{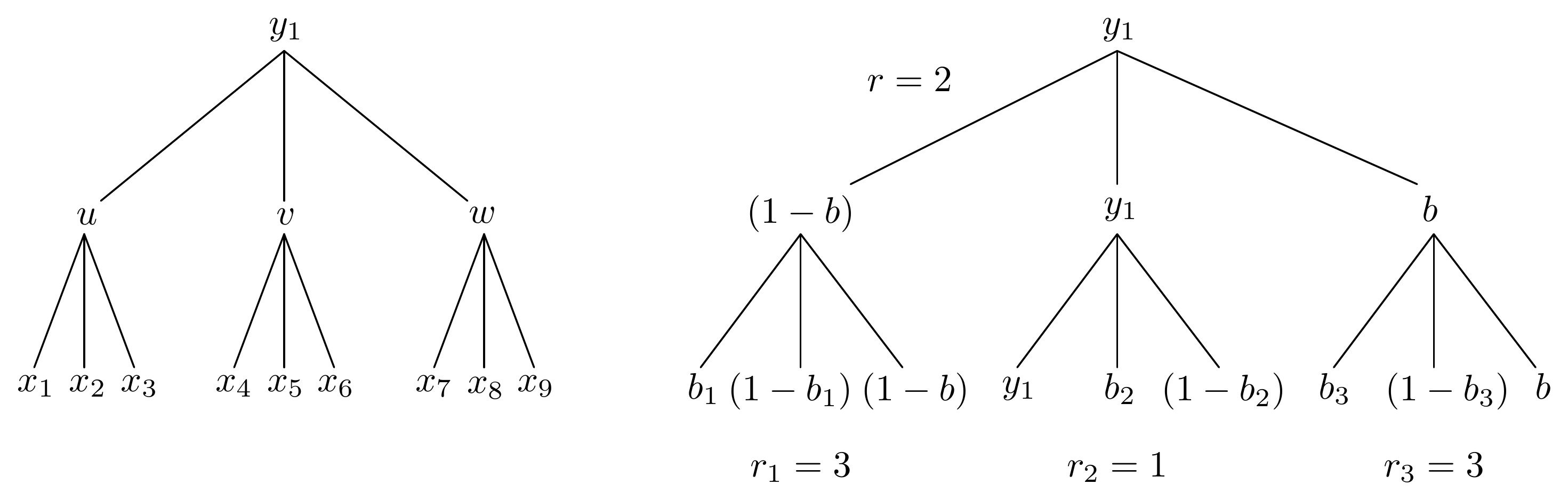}}
\caption{In a uniform 2-level encoding, a bit $y_1$ is encoded as the
height-2 recursive majority of~$9$ bits $x_1 x_2 \dotsb x_9$.
On the right hand side, we give an example of the~$9$ bits with 
specific choices of~$r,r_1,r_2,r_3$ for the two levels. In this example,
when $y_1=0$, then $b_3$ encodes the absolute minority bit
if $b= 1$ and~$b_3=0$.
\label{2-level}}
\end{figure}
This encoding is no longer a bijection. However,
one can make essentially the same argument as Theorem~\ref{thm:1level}. The
advantage of this scheme over the one used earlier is the higher
symmetry: while in the previous encoding, the instances in~$\cH_h$ are
related by cyclic permutations of triples of three siblings, now the entire 
symmetric group acts on them. Because of this higher symmetry if one of 
three siblings has been queried, the remaining two still play symmetric roles.

We later use the following observations that hold for all $h\ge k\ge1$:
\begin{enumerate}
\item For all $y\in\cH_{h-k}$ and $r\in{\cal R}^{(k)}_h$ we have
$\TMAJ_{h-k}(y) = \TMAJ_h(\psi^{(k)}(y,r))$.
\item For $(y,r)\in_\rR\cH_{h-k}\times\mathcal R^{(k)}_h$ the
value~$\psi^{(k)}(y,r)$ is distributed uniformly in ${\cal H}_{h}$.
\item For each $r\in\mathcal R^{(k)}_h$ and index $i$ in the range $1\le
i\le3^{h-k}$ there is a unique index $q_i(r)$ in the range $(i-1)3^k+1\le
q_i(r)\le i3^k$
such that for all $y\in\cH_{h-k}$ we have $x_{q_i(r)}=y_i$, where
$x=\psi^{(k)}_h(y,r)$. If $1\le j\le3^h$ but $j$ is not equal to $q_i(r)$
for any $i$, then $x_j$ is independent of the choice of $y$. We call these
bits of $x$ the {\em fixed bits}.
\end{enumerate}

We use the uniform $k$-level encoding schemes to obtain better bounds on
$p_h^\delta$. The argument is very similar to the argument in
Theorem~\ref{thm:1level}. We
start with proving a lower bound on $p^\delta_h$ based on a parameter
computable by considering all the (finitely many) decision trees
acting on inputs from $\cH_k^0$. Then we proceed to computing
this parameter. The high symmetry helps in reducing the number of cases 
to be considered, but as $k$ grows the
length of the calculation increases rather rapidly. We explain the
basic structure of the calculation and also include a short Python
program implementing it in Appendix~\ref{code}. 
As an illustration, we do the calculation for~$k= 2$ without the use
of a computer. For $k=3,4$ we include the results of the program.
A much more efficient algorithm would be needed to make the calculation 
for $k=5$ feasible.

Let us fix $k\ge1$ and let $C$ be a deterministic decision tree
algorithm on inputs of length $3^k$ that queries at least one
variable. We define
\[
\alpha_C \quad = \quad 
    \frac{ \Pr_{{x\in_\rR \cH_k^0, (y, r)\in_\rR \psi^{-1}(x)}}[C(x)
               \text{ queries }x_{q_1(r)}]}
         {\Pr_{x\in_\rR \cH_k^0}[C(x)\text{ queries } x_{m(x)}]}\enspace,
\]
where $\psi=\psi^{(k)}_k$. Since~$C$ queries at least one bit,
neither the numerator nor the denominator is zero. So
$\alpha_C$ is well defined and positive. We emphasize that~$\alpha_C$
does not depend on the output of $C$, it depends only on which input
bits $C$ queries. We further define
\[
\alpha_k \quad = \quad \max_C \; \alpha_C \enspace,
\]
where the maximum extends over all deterministic decision trees~$C$ on
$3^k$ variables that query at least one variable.

\begin{theorem}
\label{thm:klevel}
For every $k\ge1$, $h\ge0$ integers and $\delta\ge0$ real,  we have
\[
p^\delta_h \quad \geq \quad
    (1-2\delta) \left(\frac{\alpha_k}{2^k}\right) \, \alpha_k^{-h/k} \enspace.
\]
Therefore,
\[
\rR_\delta(\TMAJ_h) \quad \geq \quad
    (1-2\delta) \, \left( \frac{\alpha_k}{2^k} \right)\,
    \left( 2 + \alpha_k^{-1/k} \right)^h \enspace.
\]
\end{theorem}
\begin{proof}
We concentrate on the proof of the first statement; the second follows
from Corollary~\ref{cor: jks}.

The proof follows the same structure as that of Theorem~\ref{thm:1level}
but instead uses a depth-$k$ recursion: we show that 
\begin{equation}
\label{k-rek}
\alpha_k \, p_h^\delta \quad \ge \quad p_{h-k}^\delta
\end{equation}
if $h\ge k$. To bound $p_h^\delta$ in the base cases $h<k$, we
invoke Theorem~\ref{thm:1level}, i.e., $p_h^\delta\ge(1-2\delta)/2^h$,
and that $\alpha_k \leq 2^k$ (from the proof of the theorem).

It remains to prove Eq.~(\ref{k-rek}).
We proceed as in Theorem~\ref{thm:1level}: we consider a randomized
  $\delta$-error algorithm $A$ for $\TMAJ_h$ that achieves the minimum
  defining $p^\delta_h$ and construct a randomized
  algorithm $A'$ for $\TMAJ_{h-k}$ with the same error that queries the
  absolute minority of a uniformly random element of $\cH_{h-k}$ with
  probability at most $\alpha_k \, p^\delta_h$

To define $A'$, we use the uniform $k$-level encoding scheme $\psi=\psi^{(k)}$
(see Figure~\ref{2-level} for an illustration of the $k=2$ case). On 
input $y\in \cH_{h-k}$ the algorithm $A'$ picks a uniformly random 
element $r$ of ${\cal R}^{(k)}_h$ and simulates the decision tree
$A$ on input $x=\psi(y,r)$. Whenever~$A$ queries a fixed bit of $x$, 
the algorithm $A'$ makes no query, and when~$A$ queries a bit $x_{q_i(r)}$, 
$A'$ queries $y_i$.
Define $\cH_h'=\mathcal{H}_{h-k} \times {\cal R}^{(k)}_h$. Then $(y, r)\in
\cH_h'$ encodes $\psi(y,r)\in\cH_h$.

We partition ${\cal H}_h$ into equivalence classes, this time into sets
of size $3^l$ with $l=\sum_{i=0}^{k-1} 3^i$.  For $h=k$, the two classes 
are ${\cal H}_k^0$ and ${\cal H}_k^1$. For $h > k$, an equivalence class 
consists of inputs that are identical everywhere except in the
height-$k$ subtree containing their absolute minority. More formally,
recall that~$P(i)$ denotes the parent of a node~$i$ in a tree, and
let~$P^{(k)}$ denote the $k$-fold composition of~$P$ with itself.
In other words, $P^{(k)}(i)$ is the ancestor of the node~$i$ that 
is~$k$ levels above~$i$. The partition of~${\cal H}_h$ for~$h > k$
consists of the equivalence classes of the relation
defined as $x \sim x'$ iff~$x_i = x'_i$ for all $i$ satisfying $P^{(k)}(i) \neq
P^{(k)}(m(x))$ in the tree $\rT_h$.

Observe that the uniformity of the encoding implies that for every
equivalence class $S$, and all decision trees~$B$ in the support of $A$:
\begin{eqnarray*}
\Pr_{x\in_\rR \cH_h} [B(x)\text{ queries } x_{m(x)} \mid x\in S]
    & = & \Pr_{{(y, r)\in_\rR \cH_h', x=\psi(y, r)}}
          [B(x)\text{ queries } x_{m(x)} \mid x\in S] \enspace,
\end{eqnarray*}
and
\begin{eqnarray*}
\Pr_{(y, r)\in_\rR \cH_h'}
[B'(y,r) \text{ queries }y_{m(y)} | \psi(y, r) \in S]
    & = & \Pr_{{x\in_\rR \cH_h, (y, r)\in_\rR \psi^{-1}(x)}}
          [B'(y,r) \text{ queries }y_{m(y)} |x \in S] \enspace,
\end{eqnarray*}
where $B'$ is the algorithm that first computes $x = \psi(y, r)$
and then evaluates $B(x)$. We then prove that for every
equivalence class $S$, and all $B$ in the support of $A$, it holds that:
\begin{equation}
\label{eq:ingiaqp}
\alpha_k \Pr_{x\in_\rR \cH_h}[B(x)\text{ queries } x_{m(x)} \mid x\in S] 
    \quad \geq \quad
    \Pr_{(y, r)\in_\rR \cH_h'}
    [B'(y,r) \text{ queries }y_{m(y)} \mid \psi(y, r) \in S] \enspace.
\end{equation}
Proving Eq.~(\ref{eq:ingiaqp}) for all $B$ and $S$ finishes the proof 
of the theorem.

Let us fix $S$ and let $z$ be the undetermined part of the input, i.e., the
$3^k$ variables in the height-$k$ subtree containing the absolute
minority. Note that the set of possible values of $z$ is either
$\cH_k^0$ or $\cH_k^1$, depending on $S$. Now a deterministic decision
tree $B$ on inputs from $S$ can be considered a
deterministic decision tree $C$ for $z$. Indeed, the queries $B$ asks
outside $z$ have a deterministic answer in $S$ that can be hard wired
in $C$. In case $C$ asks no
queries at all, then Eq.~\ref{eq:ingiaqp} is satisfied with zero on
both sides of the inequality. Otherwise, if the possible
values of $z$ come from $\cH_k^0$, Eq.~(\ref{eq:ingiaqp}) follows
from $\alpha_C\le\alpha_k$ (which, in turn, comes from the definition
of $\alpha_k$ as a maximum). Finally if the possible values of $z$ are
the 1-hard inputs, Eq.~(\ref{eq:ingiaqp}) is satisfied by symmetry.
\end{proof}

To apply Theorem~\ref{thm:klevel} we need to compute (or estimate)
$\alpha_k$. For any fixed $k$ this is a finite computation, but it is
infeasible to do this by enumerating over all possible decision
trees~$C$ over~$3^k$ variables, even for small values of $k$.

For a fixed integer~$k \ge 1$ and real~$\alpha \ge 0$, we introduce a
function~$\rho_\alpha$ on decision trees on~$3^k$ variables:
\begin{equation}
\label{eq:nboqwihgaP}
\rho_\alpha(C) \quad = \quad
    \Pr_{x \in_\rR {\cal H}_k^0, (y, r)\in_\rR \psi^{-1}(x)}
    [C(x) \text{ queries }x_{q_1(r)}]
    - \alpha \Pr_{x \in_\rR {\cal H}_k^0}[C(x)\text{ queries } x_{m(x)}]
    \enspace.
\end{equation}
For the decision tree $C_0$ not querying any variables we have
$\rho_\alpha(C_0)=0$, for other decision trees $C$ we have
$\rho_\alpha(C)>0$ if and only if $\alpha_C>\alpha$. Thus, we have
$\alpha_k>\alpha$ if and only if there exists $C$ with
$\rho_\alpha(C)>0$. Finding the maximum, i.e., $\max_C\rho_\alpha(C)$
therefore answers the question whether $\alpha_k>\alpha$. We now focus
on maximizing~$\rho_\alpha$ for a given pair~$k,\alpha$. The advantage
of this approach lies in the linearity of $\rho_\alpha$, in a sense
that we clarify below. This makes it easier to maximize
$\rho_\alpha(C)$ than~$\alpha_C$ itself.

Let us call a bit of the hard input $x\in{\cal H}_k$
\emph{sensitive\/} if the flipping of this input bit flips the value of
$\TMAJ_k(x)$. Note that there are exactly $2^k$ such bits for each hard
input, these are the ones where all nodes on the root to leaf path
of the ternary tree evaluate to the same value.
 
Notice that
for a fixed $x\in{\cal H}_k^0$ and  $(y, r)\in_\rR \psi^{-1}(x)$ the
position $q_1(r)$ (where the $k$-level encoding
$\psi=\psi^{(k)}$ ``hides'' the input variable $y$) is uniformly
distributed over the $2^k$
sensitive positions. Thus, we can
simplify Eq.~(\ref{eq:nboqwihgaP}) defining $\rho_\alpha$ as follows:
\begin{equation}
\label{eq:simpler}
\rho_\alpha(C) \quad = \quad 2^{-k} \, \pi_q(C) -\alpha \, \pi_m(C)
\enspace,
\end{equation}
where $\pi_q(C)$ is the expected number of
sensitive bits queried
by $C(x)$ for $x\in_\rR {\cal H}_k^0$ and $\pi_m$ is the probability
that the absolute minority bit is queried by $C$ for $x\in_\rR{\cal H}_k^0$.

At any instant during the execution of a decision tree, we can partition
the input variables into those that have been queried and those that
have not. We call the set of pairs~$(x_i,a_i)$ of variables that have 
already been queried, along with their values, the {\em configuration\/}
of the decision tree at that instant.
The next action of the decision tree is either to stop (and produce an output
that is not relevant to this analysis) or to choose a variable that has not
yet been queried, and query it. In the latter case, the configuration 
after the query is determined by the value of the chosen variable.

A decision tree is determined by the actions it takes in the possible
configurations. In a configuration $\gamma$, a decision
tree that maximizes~$\rho_\alpha$ takes an action that maximizes the
linear combination in Eq.~(\ref{eq:simpler}) {\em conditioned on
reaching this configuration\/}. Namely it maximizes
\begin{equation}
\rho_\alpha(C,\gamma) \quad = \quad
    2^{-k} \, P_q(C,\gamma)- \alpha \, P_m(C,\gamma)
\enspace,
\end{equation}
where $P_q(C,\gamma)$ is the expected number of
sensitive bits
queried by $C(x)$, when $x$ is a uniformly random 0-hard $x$
consistent with $\gamma$, while $P_m(C,\gamma)$ is the probability that
$C(x)$ queries the absolute minority bit for a uniformly random 0-hard
$x$ consistent with $\gamma$. The optimal action in a configuration $\gamma$
can therefore be found independently of the actions taken at configurations 
inconsistent with $\gamma$. (A similar statement for the 
maximization of $\alpha_C$ is false.)

Note that $\rho_\alpha(C,\gamma)$ is easy
to compute if $C$ stops at $\gamma$. If $C$ queries a new variable
at $\gamma$, then $\rho_\alpha(C,\gamma)$ is given by a convex
combination of~$\rho_\alpha(C,\gamma')$ and~$\rho_\alpha(C,\gamma'')$, where
$\gamma'$ and $\gamma''$ are the two configurations resulting from 
the query.
This leads to the following dynamic programming algorithm:
consider all configurations in an order in which evaluating further
variables yields configurations considered earlier. For each
configuration $\gamma$ we iterate through all actions to find an
optimal one and store the value
of $\rho_\alpha(C,\gamma)$ for the optimal $C$. We have
$\rho_\alpha(C)=\rho_\alpha(C,\emptyset)$, where $\emptyset$ is the
initial configuration (in which no variable has been queried).

The number of all possible configurations is $3^{3^k}$. This makes 
the above algorithm infeasible even for $k=3$. We reduce the number of
configurations considered significantly by appealing to simple 
properties of $0$-hard inputs:

\begin{enumerate}
\item We use the symmetries of the hard distribution. All the 
configurations~$\gamma$
in an orbit generated by the automorphisms of the ternary tree give rise
to the same value for~$\rho_\alpha(C,\gamma)$.
We consider only one configuration in each such equivalence class.

\item \label{item-unstable}
We single out two types of configurations in which an optimal action
is clear without the need for further computation. First,
if the configuration uniquely determines the value of the root (i.e., is
not consistent with a~$1$-input),
an optimal strategy is to stop. Second, if an unqueried variable is 
known not to be the
  absolute minority variable, querying it does not decrease the
objective function. We may thus assume that an optimal decision tree
queries this variable.

We implement the second type of action as follows.
The nodes in the path to the absolute minority evaluate
  alternately to 0 and 1. If a node at odd depth evaluates to 0 or
  at even depth evaluates to 1, then it is not on the absolute
  minority path. In this case, all variables in the subtree rooted at the node 
are queried by an optimal decision tree. If a
  node at odd depth evaluates to 1, or at even depth evaluates to 0,
  then its {\em siblings\/} are not on the absolute minority path. In
this case, the variables in the subtrees rooted at the siblings are queried
by an optimal decision tree.
\end{enumerate}

We call a configuration~$\gamma$ {\em unstable\/} if there is an action 
that an optimal decision tree may take in~$\gamma$ as described in 
point~\ref{item-unstable}. We call the configuration \emph{stable\/} 
otherwise. Note that the value of the root is not uniquely determined by
a stable configuration.
It suffices to store $\rho_\alpha(C,\gamma)$ for stable $\gamma$. If
$\rho_\alpha(C,\gamma)$ is needed for some unstable configuration
$\gamma$ we apply the above rules (possibly multiple times) until the
value of the root is determined, or we arrive at a stable configuration.
We then compute $\rho_\alpha(C,\gamma)$ using the appropriate stored 
values.

The lone stable configuration for height $0$ is~$\emptyset$,
the one in which the input variable has not been queried.
Consider a stable configuration for height-$k$ formulae, for~$k \ge 1$,
and the restrictions of the configuration to the subtrees rooted at
the three children of the root. Call the configuration obtained by
negating the values of the variables in a configuration its
\emph{dual\/}. It is straightforward to verify that either a 
restriction uniquely determines the value of the corresponding child, 
or it is a dual of a stable configuration for height~$(k-1)$.
No child of the root can be known to have value~$1$ in a stable
configuration and at most one of them can be known to have value~$0$.
So an equivalence class of stable configurations for height~$k$ is 
determined by a multiset of (equivalence classes of) stable 
configurations for height~$(k-1)$ of size~$2$ (when one child 
is known to have value~$0$) or of size~$3$ (when the values of the children
are all undetermined). This characterization gives us the following
recurrence relation for~$N_k$, the number of equivalence classes of
stable configurations for height $k$:
\[
N_k \quad = \quad {N_{k-1}+1\choose2}+{N_{k-1}+2\choose3} \enspace,
\]
with the initial condition~$N_0 = 1$. The recurrence gives us
\begin{eqnarray*}
N_1&=&2 \\
N_2&=&7 \\
N_3&=&112 \\
N_4&=&246,792 \qquad \textrm{and} \\
N_5&=&2,505,258,478,767,772 \enspace.
\end{eqnarray*}
This makes the dynamic programming approach for
optimizing~$\rho_\alpha(C)$ feasible for $k\le4$. This approach is
implemented by the Python program presented in Appendix~\ref{code}.
In order to avoid dealing with dual configurations, the Python program 
considers {\sf NOT-3-MAJ}, the recursive negated majority-of-three
function. This is the function computed by the a negated majority gate 
in every internal node of a complete ternary tree.

Recall that using an algorithm to maximize $\rho_\alpha(C)$
we can check whether $\alpha_k>\alpha$. Instead of a binary
search we find the exact value of $\alpha_k$ as follows. With little
modification, the algorithm we present for maximizing $\rho_\alpha(C)$ 
also produces the value~$\alpha_{C^*}$ for the optimal decision tree $C^*$.
We then start with an arbitrary $\alpha\le\alpha_k$, and repeatedly 
optimize~$\rho_\alpha$ updating the estimate $\alpha$ to the 
last value $\alpha_{C^*}$, until we find that $\max_C\rho_\alpha(C)=0$. 
This heuristic finds the maximum $\alpha_k$ in a finite 
number of iterations. Instead of bounding the number of iterations in 
general we mention that starting from the initial value $\alpha=0$ 
the heuristic gives us $\alpha_k$ in at most four iterations when
$k=2,3,4$. The computations show:
\begin{eqnarray*}
\alpha_1&=&2 \enspace,\\
\alpha_2&=&\frac{24}7 \enspace,\\
\alpha_3&=&\frac{12231}{2203} \enspace, \quad \textrm{and}\\
\alpha_4&=&\frac{2027349}{216164} \enspace.
\end{eqnarray*}

Using the value of $\alpha_4$, Theorem~\ref{thm:klevel} yields the
following bound.
\begin{corollary}
\label{thm:4level}
$$
\rR_\delta(\TMAJ_h)
    \quad \geq \quad (1/2-\delta) \left(2 +
                     \left( \frac{216164}{2027349}\right)^{1/4}
                     \right)^h
    \quad > \quad (1/2-\delta)2.57143^h \enspace.
$$
\end{corollary}

\subsection{Analysis of the $2$-level encoding}

As an illustration of the use of higher level encodings we explicitly 
derive a second order recurrence for~$p_h^\delta$ using $2$-level
encodings. We fix $k=2$ and consider deterministic decision trees $C$ on
$9$ variables. We run these decision trees on inputs from ${\cal
  H}_2^0$.

From the proof of Theorem~\ref{thm:1level}, we have~$\alpha_C\leq 4$ for 
all $C$. We may verify that
$\alpha_{C_0}=3$ for the decision tree $C_0$ with the following strategy:
first query $x_1$, if $x_1=1$, stop, else if $x_1=0$, query $x_2$ and $x_3$;
then if $\MAJ(x_1,x_2,x_3)=0$, stop, else query all remaining bits and
stop.
These bounds show that $3\le\alpha_2\le4$, so it suffices to consider
$\rho_\alpha$ for the values of $\alpha$ in the range $[3,4]$.
We prove below that for these values of~$\alpha$ a single decision tree~$C'$
maximizes $\rho_\alpha(C)$ among the deterministic decision trees that
query at least one variable. This decision tree $C'$ is given in
Figure~\ref{fig:2partial}.
We state the optimality of $C'$ in the following lemma.

\begin{lemma}
\label{lem:2level}
Let $C$ be any deterministic decision tree on 9-bit inputs
that makes at least one query and let
$C'$ be the decision tree depicted in Figure~\ref{fig:2partial}.
Then for all $\alpha \in [3, 4]$, $\rho_\alpha(C)\leq\rho_\alpha(C')$.
\end{lemma}
\begin{proof}
Recall that the action in a configuration $\gamma$ of the decision tree $C$ that maximizes
$\rho_\alpha(C)$  is the one that maximizes
$\rho_\alpha(C,\gamma)=2^{-2} P_q(C,\gamma)-\alpha P_m(C,\gamma)$. To
simplify notation, we write $\rho$, $P_q$ and $P_m$ for
$\rho_\alpha(C,\gamma), P_q(C,\gamma)$, and~$P_m(C,\gamma)$,
respectively, if the decision tree~$C$ and configuration $\gamma$ 
considered are clear from the context.

We call any set of
$3$ sibling nodes a {\em clause}, so that $\{1,2,3\}$, $\{4,5,6\}$ and
$\{7,8,9\}$ are clauses.
We say a clause is \emph{evaluated\/} if its majority is uniquely
determined by the configuration under consideration.

We argue that an algorithm that maximizes $\rho_\alpha(C)$ takes certain
actions, without loss of generality. We begin with three rules that are 
special cases of the general rules from Section~\ref{sec:lb} that we used 
to reduce the number of configurations considered.
\begin{enumerate}
\item If a bit is evaluated to~$0$, then evaluate all remaining bits 
in its clause.
\item If two bits in a clause are evaluated to $1$ (this is the minority clause), evaluate all remaining bits in the other clauses.
\item if two clauses have been evaluated to $0$, then stop.
\end{enumerate}

In what follows we systematically consider all stable configurations for
height~$2$ inputs, i.e., the ones in which the above three rules do not apply. 
For each such configuration, we determine what action(s) an optimal decision 
may take next, without loss of generality, in order to maximize $\rho_\alpha$.
\begin{enumerate}
\setcounter{enumi}{3}
\item \label{item:singlemaj1} A single majority ($0$) clause is
  evaluated and either no variables are evaluated in either of the
  other clauses or a single $1$ is evaluated in both the other
  clauses.  In this case, stopping is the best strategy.  Indeed,
  $m(x)$ has not been queried yet. Therefore if we stop,
  then $P_m=0$ and $P_q =2$. This gives $\rho=1/2$.
  But if $C$ continues by querying at least one more
  bit, then $P_m \ge 1/6$ or $P_m \ge 1/4$ (since there are either $6$ or
  $4$ remaining unqueried variables, respectively, and they are symmetric) and
  $P_q\leq 4$.  Therefore, $\rho=2^{-2}P_q-\alpha
  P_m\le1-\alpha/6\le1/2$ since $\alpha\geq 3$.

\item A single majority clause is evaluated and one more bit is evaluated to
  $1$, but nothing more.  There are $9$ inputs consistent with this 
configuration.  We argue that stopping is best,
  as in the previous case. The argument is more involved because
  there is no global symmetry between the unqueried variables. We 
  separately compare stopping with querying a variable inside or outside 
the untouched clause.

 If we stop, then
 $P_q = 2$ and $P_m = 0$, so we have $\rho= 1/2$.

  If we query a variable in the clause containing the evaluated bit~$1$,
  then there are $3$ consistent inputs in which the next queried bit is
  $m(x)$. So we have $P_m\ge1/3$ and $P_q\le4$. Since $\alpha \ge 3$,
$\rho \le0$.

  If we query a variable in the untouched clause, then there is $1$
  out of the $9$ consistent inputs for which this next queried
  variable is $m(x)$, making $P_m\ge1/9$. There are $4$ more
  consistent inputs for which this variable evaluates to $1$.  In this
  case we arrive at the configuration covered by item \ref{item:singlemaj1}
above. Using that rule, the algorithm stops, leaving 2 out the 4
sensitive bits unqueried. Thus, we have $P_q\le4-4/9\cdot2=28/9$ and
  $\rho\le 4/9$, which is lesser than the $1/2$ obtained if we stop.

\item A single $1$ has been evaluated in each of the three clauses and
  no other bit has been queried.  In this case reading another bit is
  the best strategy (the choice of which bit is unimportant because of
  symmetry).  Observe that no bit $0$ has been evaluated yet.
  Therefore if $C$ stops, we have $\rho=P_q=P_m=0$. If $C$ continues to query
  another bit, which bit it queries does not matter by symmetry. Then
the rest of
  the algorithm is determined by the earlier rules yielding $P_q=8/3$,
  $P_m=1/6$, and $\rho=2/3-\alpha/6\geq0$.

\item A single bit $1$ has been evaluated in each of two different
  clauses and the third clause is untouched. Then it is best to
  evaluate a bit of the third clause.  If $C$ stops we have $\rho=P_m=P_q=0$.

 If $C$ reads
  another bit in one of the clauses containing a single $1$ bit, then
  the rest of the decision tree algorithm is determined by the earlier rules
  and we get $P_q=14/5$ and $P_m=1/5$ with $\rho
  =7/10-\alpha/5$. Note that whether this option is better than
  stopping depends on the value of $\alpha$.
  
  If $C$ reads a bit in the untouched clause, then by using the rules
  already presented in the previous cases, we calculate $P_q=12/5$ and
  $P_m=2/15$, yielding $\rho=3/5-2\alpha/15$. This is
  more than both 0 and $7/10-\alpha/5$ in the entire
  range of $\alpha$ we are considering.

\item A single bit $1$ has been evaluated in one clause and no other
  clauses are touched. Then it is best to evaluate a bit of another
  clause.  If $C$ stops, then we have $\rho=P_m=P_q=0$.

  If $C$ evaluates another bit in the clause containing
  $1$, then the rest of $C$ is determined by earlier rules and we have
  $P_q=3$, $P_m=1/4$ and $\rho=3/4 - \alpha/4 \leq 0$.

  But if $C$ queries a bit in an untouched clause, then similar
  calculations yield $P_q=7/3$, $P_m=5/36$ and $\rho
  =7/12-5\alpha/36>0$, making this the best choice.
\end{enumerate}

Following all the above rules and always choosing the smallest index when
symmetry allows us to choose, we arrive at a well defined decision tree,
namely~$C'$. This finishes the proof of the lemma.
\end{proof}

\begin{figure}
\begin{center}
\includegraphics[width=.8\textwidth]{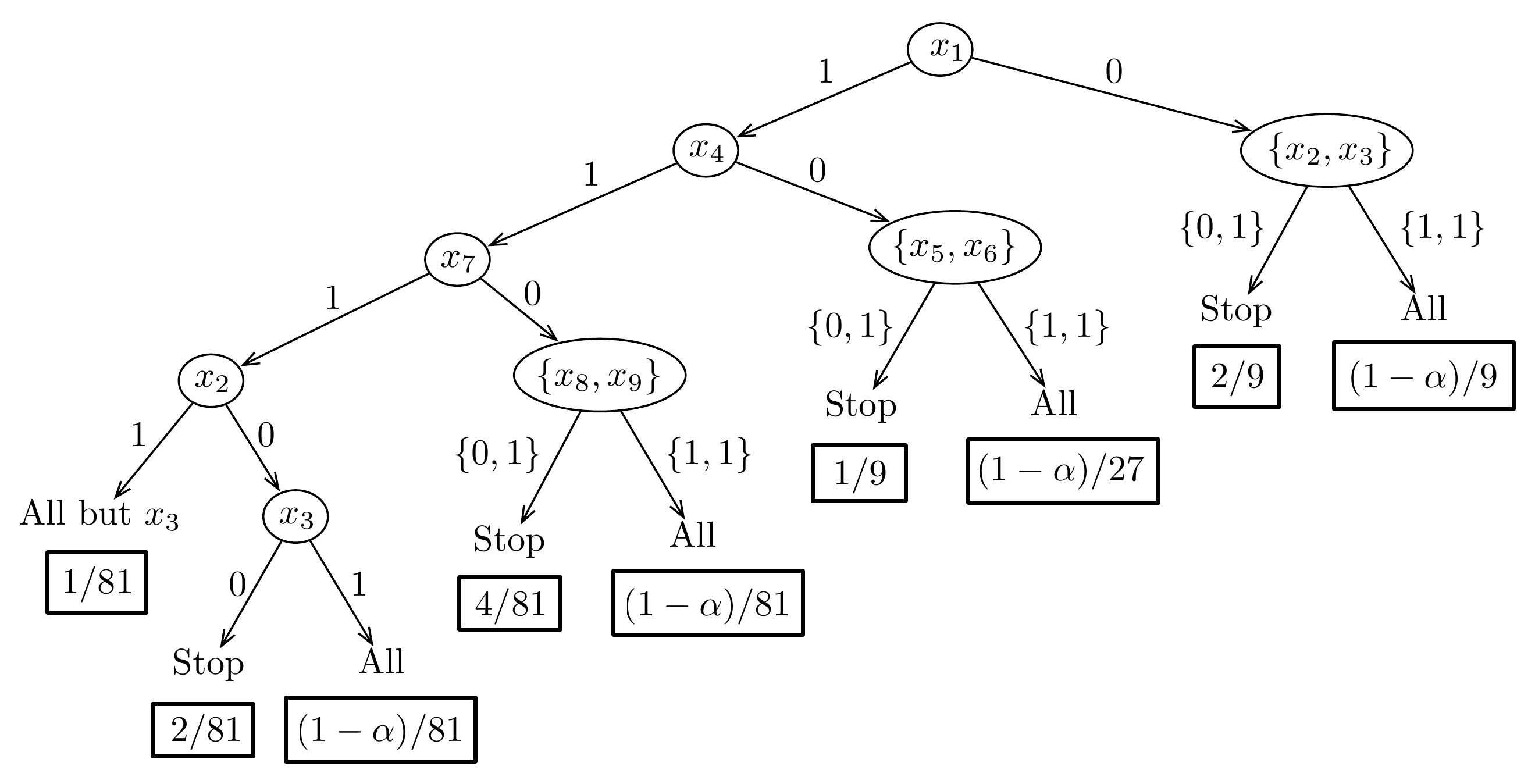}
\end{center}
\caption{A pictorial representation of the decision tree $C'$.
``Stop'' indicates that the algorithm stops and produces an output (that
is not relevant to our analysis). ``All'' indicates that the algorithm
queries all remaining variables. In the leftmost branch, 
all variables but~$x_3$ are queried.
The contribution to $\rho_{\alpha}(C')$ by each branch of $C'$ is written
under it in a box.
\label{fig:2partial}
}
\end{figure}

The above lemma immediately gives us the value of~$\alpha_2$.

\begin{theorem}
$\alpha_2=24/7$.
\end{theorem}

\begin{proof}
As observed in the paragraph before Lemma~\ref{lem:2level} we have
$3\le\alpha_2\le 4$. By the lemma we know that $\rho_\alpha(C)$ in
this range is maximized by
either $C'$ or the decision tree that does not query any variable. The
latter gives $\rho_\alpha=0$. We have $\alpha \ge  \alpha_2$ if and 
only if this maximum is $0$, so we are done
if we calculate $\rho_\alpha(C')$. The contribution of each branch of
the algorithm is given in Figure~\ref{fig:2partial}. Together these sum to
$(48-14\alpha)/81$. This is positive for $\alpha<24/7$, so we have
$\alpha_2=\alpha_{C'}=24/7$.
\end{proof} 

As a consequence of the value $\alpha_2$, we get a slightly weaker 
lower bound than the one in Corollary~\ref{thm:4level}.
However, this bound has the advantage that we have an explicit proof for it.

\begin{corollary}
$\rR_\delta(\TMAJ_h) \quad \geq \quad (1/2-\delta)
\left( 2+\sqrt{7/24\,}\right)^h \quad > \quad (1/2-\delta) \, 2.54006^h$.
\end{corollary}

\section{Improved Depth-Two Algorithm}
\label{sec-algo}

In this section, we present a new zero-error algorithm for
computing~$\TMAJ_h$. For the key ideas behind it, we refer the reader to 
Section~\ref{sec-introduction}.

As before, we identify the formula~$\TMAJ_h$ with a complete ternary
tree of height~$h$.  We are given an assignment to the~$3^h$ variables
(equivalently, the leaves) of~$\TMAJ_h$, which may be accessed by
querying the variables. In the description of the algorithm we adopt the
following convention.  Once the algorithm has determined the value~$b$
of the subformula rooted at a node~$v$ of the formula~$\TMAJ_h$, we also
use~$v$ to denote this bit value~$b$.

The algorithm is a combination of two depth-2 recursive
algorithms. The first one, $\evaluate$ (see \algoone{}), takes a node $v$ of height
$h(v)$, and evaluates the subformula rooted at~$v$.  The interesting
case, when $h(v) > 1$, is depicted in Figure~\ref{fig-evaluate}.
The first step, permuting the input, means
applying a random permutation to the children $y_1, y_2, y_3$ of $v$
and independent random permutations to each of the three sets of
grandchildren.
\begin{algorithm} \footnotesize
\caption{\evaluate($v$): evaluate a node~$v$.}
\label{alg:algoone}
\begin{algorithmic}
\Input{Node~$v$ with subtree of height $h(v)$.}
\Output{the bit value~$\TMAJ_h(Z(v))$ of the
subformula rooted at~$v$}
\BlankLine
\State Let $h = h(v)$
\BlankLine
\If{$h = 0$} \Comment{First base case: $h = 0$ ($v$ is a leaf)}
    \State Query~$Z(v)$ to get its value~$a$;
    \Return $a$
\EndIf
\BlankLine

\State Let~$y_1, y_2, y_3$ be a uniformly random permutation 
of the children of~$v$ \Comment{$v$ has height $h \ge 1$}
\BlankLine

\If{$h = 1$}  \Comment{Second base case: $h = 1$}
    \State $\evaluate(y_1)$ and $\evaluate(y_2)$
    \If{$y_1 = y_2$}
        \Return $y_1$
    \Else\ 
        \Return $\evaluate(y_3)$
    \EndIf
\EndIf

\State \Comment{Recursive case: $v$ has height~$h \ge 2$; use the
attached figure as a guide}

\State Let~$x_1$ and~$x_2$ be chosen uniformly at random from the
children of~$y_1$ and~$y_2$, respectively
\BlankLine \\
\hspace*{8.5cm}\includegraphics[height=100pt]{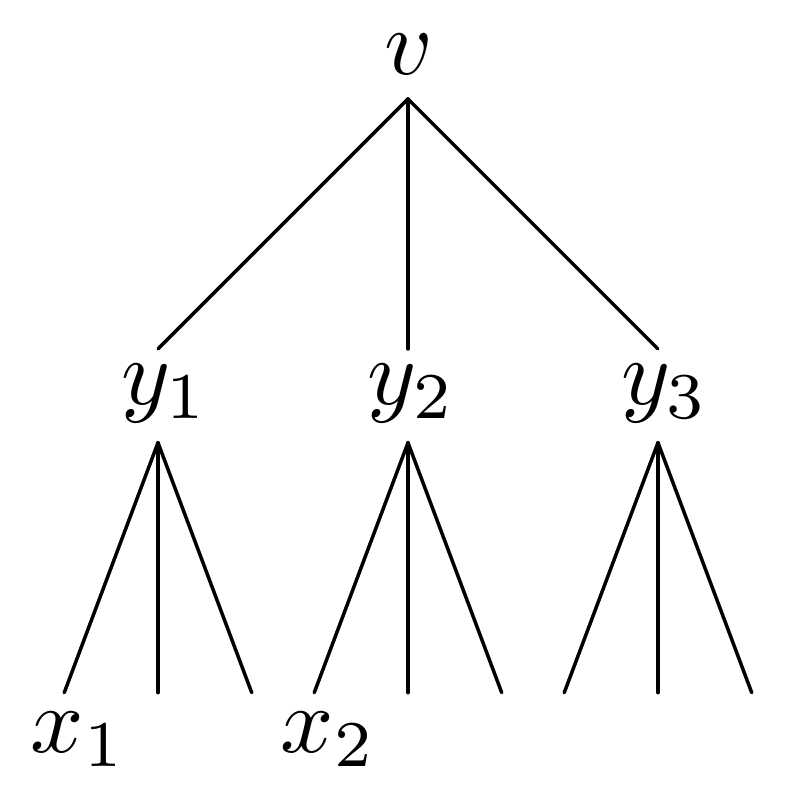}\vspace*{-3.5cm}

\State $\evaluate(x_1)$ and~$\evaluate(x_2)$
\BlankLine

\If{$x_1 \neq x_2$}
    \State $\evaluate(y_3)$
    \State Let $b \in \set{1,2}$ be such that $x_b=y_3$
    \State $\complete(y_b, x_b)$
    \If{$y_b=y_3$} 
        \Return  $y_b$ 
    \Else\ 
        \Return $\complete(y_{3-b}, x_{3-b})$
    \EndIf
\Else\ [$x_1=x_2$]
    \State $\complete(y_1, x_1)$
    \If{$y_1=x_1$}
        \State $\complete(y_2, x_2)$
        \If{$y_2 = y_1$}
            \Return $y_1$
        \Else\ [$y_2 \neq y_1$]
            \Return $\evaluate(y_3)$
        \EndIf
    \Else\ [$y_1 \neq x_1$]
        \State $\evaluate(y_3)$
        \If{$y_3 = y_1$}
            \Return $y_1$
        \Else\ 
            \Return $\complete(y_2, x_2)$
        \EndIf
    \EndIf
\EndIf
\end{algorithmic}
\end{algorithm}
\begin{figure}
\begin{center}
\includegraphics[width=.95\textwidth]{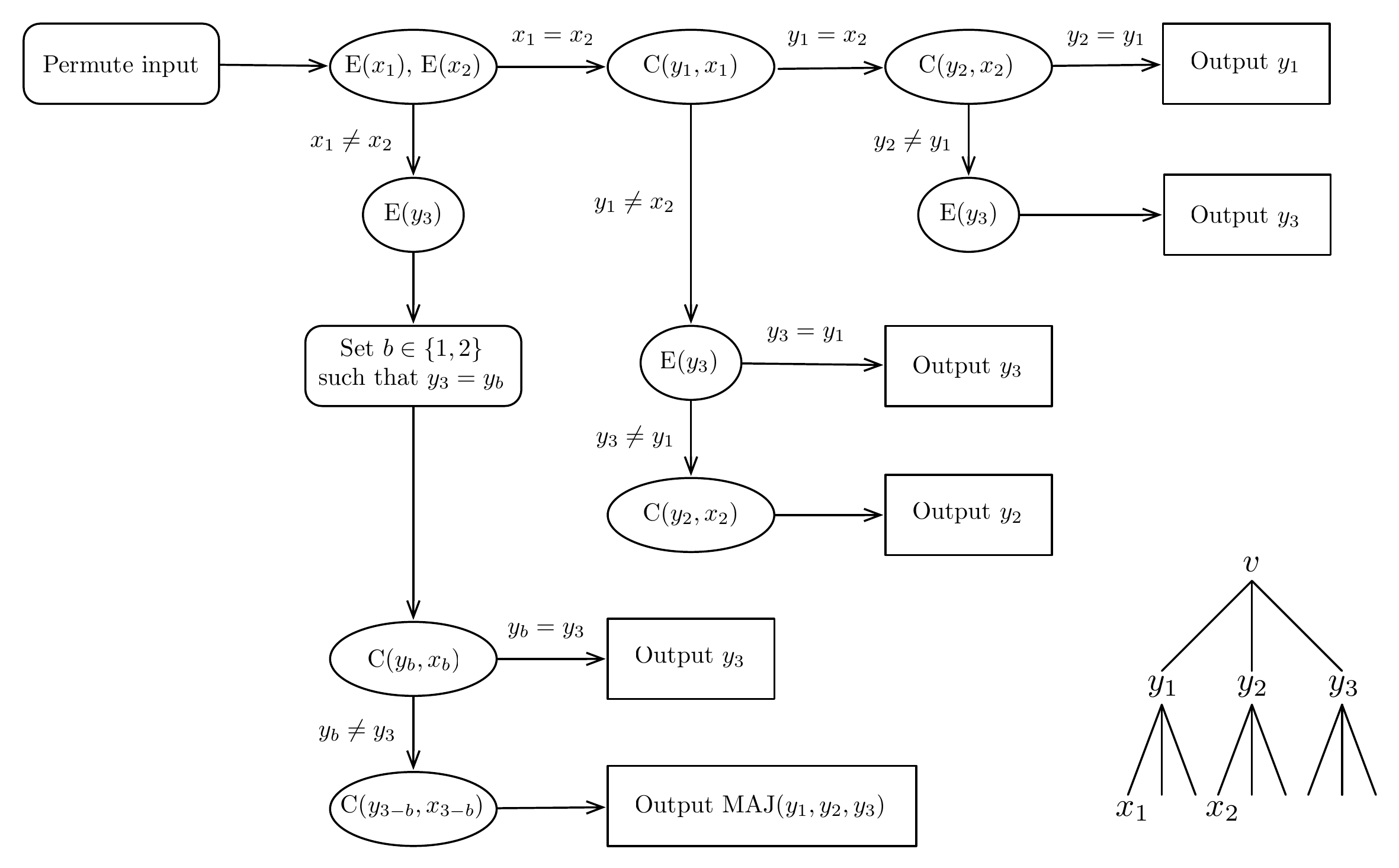}\hfil
\end{center}
\caption{Pictorial representation of algorithm $\evaluate$ on a
  subformula of height~$h(v) \geq 2$ rooted at~$v$. It is abbreviated
  by the letter `E' when called recursively on descendants of~$v$. The
  letter `C' abbreviates the second algorithm $\complete$ depicted in
Figure~\ref{fig-complete}. }
\label{fig-evaluate}
\end{figure}

The second algorithm, $\complete$ (see \algotwo{}), is depicted in
Figure~\ref{fig-complete}.
It takes two arguments~$v,y_1$, and
completes the evaluation of the subformula~$\TMAJ_h$ rooted at node~$v$,
where~$h(v) \geq 1$, and~$y_1$ is a child of~$v$ whose value has
already been evaluated.  The first step, permuting the input, means
applying a random permutation to the children $y_2, y_3$ of $v$ and
independent random permutations to each of the two sets of
grandchildren of $y_2, y_3$.  Note that this is similar in form to the
depth $2$ algorithm due to~\cite{JaKS03}.
\begin{algorithm} \footnotesize
\caption{\complete($v,y_1$): finish the evaluation of the 
subformula rooted at node~$v$}
\label{alg:algotwo}
\begin{algorithmic}
\Input{Node~$v$ of height $h(v)$; child~$y_1$ of~$v$ which has already
  been  evaluated}
\Output{the bit value~$\TMAJ_h(Z(v))$}
\BlankLine

\State Let $h = h(v)$
\BlankLine

\State Let~$y_2, y_3$ be a uniformly random permutation of the two
children of~$v$ other than~$y_1$ 
\BlankLine

\If{$h = 1$}\Comment{Base case}
    \State $\evaluate(y_2)$ 
    \If{$y_2 = y_1$}
        \Return $y_1$
    \Else\ 
        \Return $\evaluate(y_3)$
    \EndIf
\EndIf
\BlankLine

\State Let~$x_2$ be chosen uniformly at random from the
children of~$y_2$ \Comment{Recursive case}
\State \Comment{use the attached figure as a guide}\\
\hspace*{8.5cm}\includegraphics[height=100pt]{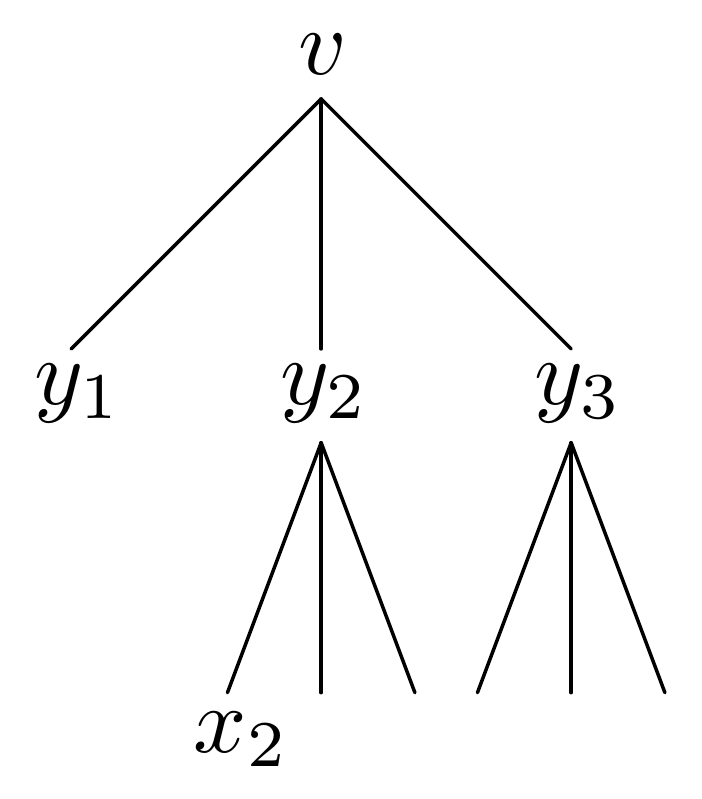}\vspace*{-3.5cm}
\State  $\evaluate(x_2)$
\BlankLine

\If{$y_1\neq x_2$}
    \State $\evaluate(y_3)$
    \If{$y_1=y_3$} 
        \Return  $y_1$ 
    \Else\ 
        \Return  $\complete(y_2, x_2)$
    \EndIf
\Else\ [$y_1=x_2$]
    \State $\complete(y_2, x_2)$
    \If{$y_1=y_2$}
        \Return  $y_1$ 
    \Else\ 
        \Return  $\evaluate(y_{3})$
    \EndIf
\EndIf
\end{algorithmic}
\end{algorithm}
\begin{figure}
\begin{center}
\includegraphics[width=.9\textwidth]{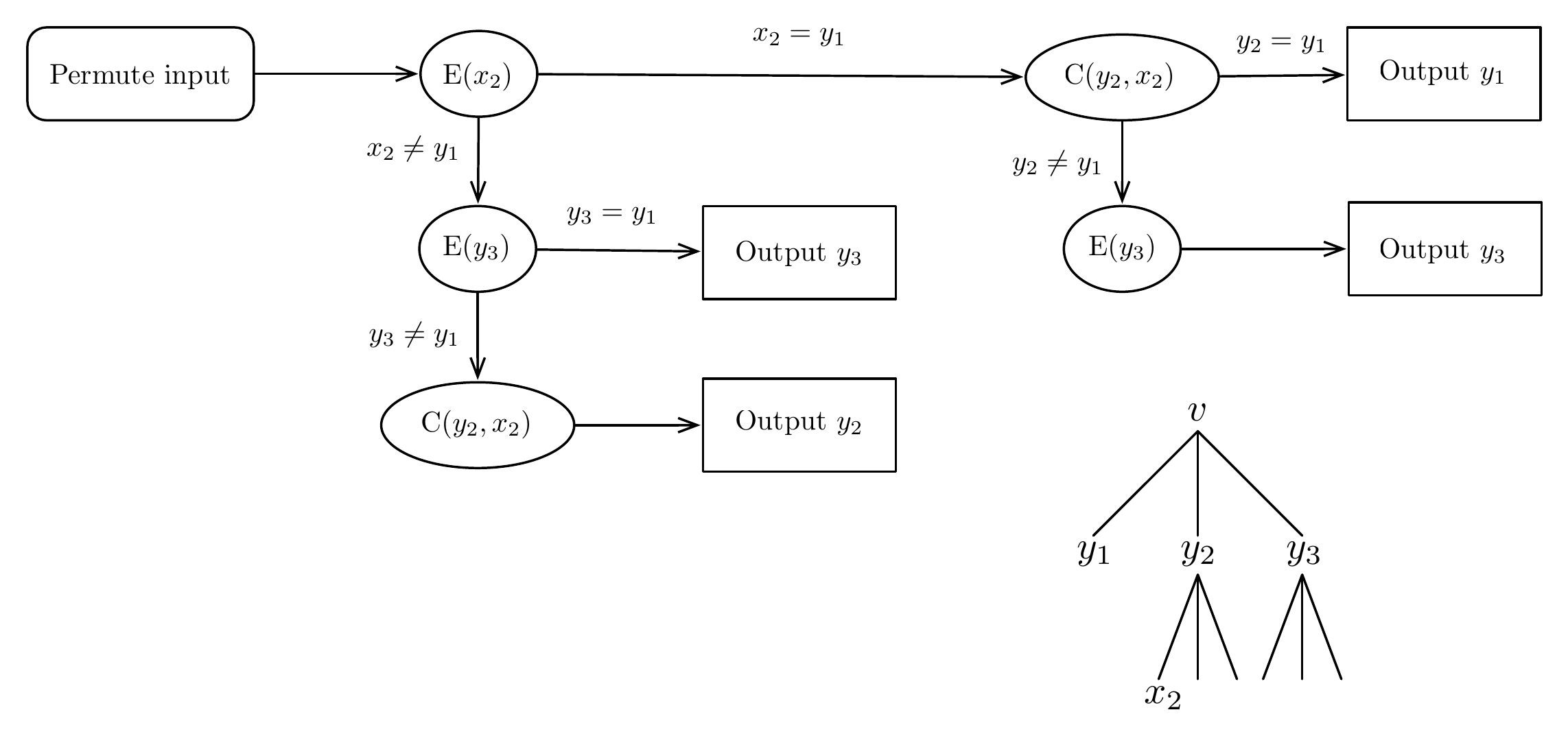} 
\end{center}
\caption{Pictorial representation of algorithm $\complete$ on a
  subformula of height~$h \geq 1$ rooted at~$v$ one child~$y_1$ of
  which has already been evaluated. It is abbreviated by the letter
  `C' when called recursively on descendants of~$v$.  Calls to
  $\evaluate$ are denoted `E'.}
\label{fig-complete}
\end{figure}

To evaluate an input of height $h$, we invoke~$\evaluate(r)$, where~$r$
is the root.  The correctness of the two algorithms follows by
inspection---they determine the values of as many children of the
node~$v$ as is required to compute the value of~$v$.

For the complexity analysis, we study the expected number of queries
they make for a worst-case input of fixed height $h$. 
(\emph{A priori\/}, we do not know if such an input is a hard input 
as defined in Section~\ref{sec-hard-distr}.)
Let $T(h)$ be
the worst-case complexity of $\evaluate(v)$ for $v$ of height $h$.  For
$\complete(v,y_1)$, we distinguish between two cases.  Let $y_1$ be
the child of node~$v$ that has already been evaluated. The complexity
given that $y_1$ is the minority child of~$v$ is denoted by~$S^{\m}$,
and the complexity given that it is a majority child is denoted
by~$S^{\rM}$.

The heart of the analysis is the following set of recurrences
that relate~$T, S^{\rM}$ and~$S^{\m}$ to each other.
\begin{lemma}
\label{lemma:algoonetwo}
We have
$S^{\m}(1) = 2$, 
$S^{\rM}(1) = \frac{3}{2}$,
$T(0) = 1$, 
and
$T(1) = \frac{8}{3}$.

For all~$h \geq 1$, it holds that
\begin{equation}
\label{eqn-SMbd}
S^{\rM}(h) \leq S^{\m}(h) \qquad \textrm{and} \qquad 
S^{\rM}(h) \leq T(h) \enspace.
\end{equation}
Finally, for all~$h \geq 2$, it holds that
\begin{eqnarray}
\label{eqn-Sm}
S^{\m}(h)  & = & T(h-2) + T(h-1) + \frac{2}{3}\, S^{\rM}(h-1) +     
                    \frac{1}{3}\, S^{\m}(h-1) \enspace, \\
\label{eqn-SM}
S^{\rM}(h) & = & T(h-2) + \frac{2}{3}\, T(h-1) 
                    + \frac{1}{3}\, S^{\rM}(h-1) + \frac{1}{3}\, S^{\m}(h-1)
\enspace,
                    \quad \textrm{and} \\
\label{eqn-T}
T(h)      & = & 2\, T(h-2) + \frac{23}{27}\, T(h-1) + \frac{26}{27}\, 
                S^{\rM}(h-1) + \frac{18}{27}\, S^{\m}(h-1) \enspace.
\end{eqnarray}
\end{lemma}
\emph{Proof.}
  We prove these relations by induction.  The bounds for~$h \in
  \set{0, 1}$ follow immediately by inspection of the algorithms.  To
  prove the statement for $h \geq 2$, we assume the recurrences hold
  for all~$l < h$.  Observe that it suffices to prove Equations
  (\ref{eqn-Sm}), (\ref{eqn-SM}), (\ref{eqn-T}) for height $h$, since
  the values of the coefficients immediately imply that
  Inequalities~(\ref{eqn-SMbd}) holds for $h$ as well. 

\textbf{Equation (\ref{eqn-Sm}).}
Since $\complete(v,y_1)$ always starts by computing the value of a
grandchild~$x_2$ of~$v$, we get the first term $T(h-2)$ in
Eq.~(\ref{eqn-Sm}).  It remains to show that the worst-case complexity
of the remaining queries is $T(h-1) + (2/3)S^{\rM}(h-1) +
(1/3)S^{\m}(h-1)$.

Since~$y_1$ is the minority child of~$v$, we have that $y_1\neq
y_2=y_3$. The complexity of the remaining steps is summarized in the
next table in the case that the three children of node $y_2$ are not
all equal. In each line of the table, the worst case complexity is
computed given the event in the first cell of the line.  The second
cell in the line is the probability of the event in the first cell
over the random permutation of the children of $y_2$.  This gives a
contribution of $T(h-1) + (2/3)S^{\rM}(h-1) + (1/3)S^{\m}(h-1)$.

$$
\begin{array}{|c|c|c|}
\hline\multicolumn{3}{|c|}{S^{\m}(h) \text{ (we have } y_1\neq y_2=y_3)}\\
\hline
\text{event}&\text{probability}&\text{complexity}\\
\hline
y_2=x_2&2/3&T(h-1)+S^{\rM}(h-1)\\
\hline
y_2\neq x_2&1/3&T(h-1)+S^{\m}(h-1)\\
\hline
\end{array}
$$
 
This table corresponds to the worst case, as the only other case is
when all children of $y_2$ are equal, in which the cost is $T(h-1) +
S^{\rM}(h-1)$.  Applying Inequality~(\ref{eqn-SMbd}) for~$h-1$, this
is a smaller contribution than the case where the children are not all
equal.

Therefore the worst case complexity for $S^{\m}$ is given by
Eq.~(\ref{eqn-Sm}).  We follow the same convention and appeal to this
kind of argument also while deriving the other two recurrence
relations.

\textbf{Equation (\ref{eqn-SM}).}
Since $\complete(v,y_1)$ always starts by computing the value of a
grandchild~$x_2$ of~$v$, we get the first term $T(h-2)$ in
Eq.~(\ref{eqn-SM}).  There are then two possible patterns, depending
on whether the three children~$y_1, y_2, y_3$ of~$v$ are all
equal. If~$y_1 = y_2 = y_3$, we have in the case that all children of
$y_2$ are not equal that:
$$
\begin{array}{|c|c|c|}
\hline\multicolumn{3}{|c|}{S^{\rM}(h) \text{ if } y_1= y_2=y_3}\\
\hline
\text{event}&\text{probability}&\text{complexity}\\
\hline
y_2=x_2&2/3&S^{\rM}(h-1)\\
\hline
y_2\neq x_2&1/3&T(h-1)\\
\hline
\end{array}
$$
As in the above analysis of Eq.~(\ref{eqn-Sm}), applying
Inequalities~(\ref{eqn-SMbd}) for height $h-1$ we get that the
complexity in the case when all children of $y_2$ are equal is bounded
above by the complexity when the children are not all equal.
Therefore the above table describes the worst-case
complexity for the case when $y_1 = y_2 = y_3$.

If~$y_1, y_2, y_3$ are not all equal, we have two events $y_1=y_2\neq
y_3$ or $y_1 = y_3 \neq y_2$ of equal probability as~$y_1$ is a
majority child of~$v$. This leads to the following tables for the case
where the children of $y_2$ are not all equal
$$
\begin{array}{|c|c|c|}
\hline\multicolumn{3}{|c|}{S^{\rM}(h) \text{ given } y_1= y_2\neq y_3}\\
\hline
\text{event}&\text{prob.}&\text{complexity}\\
\hline
y_2=x_2&2/3&S^{\rM}(h-1)\\
\hline
y_2\neq x_2&1/3&T(h-1)+S^{\m}(h-1)\\
\hline
\end{array}
\qquad
\begin{array}{|c|c|c|}
\hline\multicolumn{3}{|c|}{S^{\rM}(h) \text{ given } y_1 = y_3 \neq y_2}\\
\hline
\text{event}&\text{prob.}&\text{complexity}\\
\hline
y_2=x_2&2/3&T(h-1)\\
\hline
y_2\neq x_2&1/3&T(h-1)+S^{\m}(h-1)\\
\hline
\end{array}
$$
As before, we apply Inequalities~(\ref{eqn-SMbd}) for height $h-1$
to see that the worst case occurs when the children of~$y_2$ are not
all equal.

From the above tables, we deduce that the worst-case complexity occurs
on inputs where $y_1, y_2, y_3$ are not all equal. This is because
Inequalities~(\ref{eqn-SMbd}) for height $h-1$ imply that,
line by line, the complexities in the table for the case $y_1=y_2=y_3$
are upper bounded by the corresponding entries in each of the latter
two tables. To conclude Eq.~(\ref{eqn-SM}), recall that the two
events~$y_1 = y_2 \neq y_3$ and~$y_1 = y_3 \neq y_2$ occur with
probability~$1/2$ each:
\begin{eqnarray*}
S^{\rM}(h) & = & T(h-2) + \frac{1}{2} \left[ \frac{2}{3}\, S^{\rM}(h-1)
+ \frac{1}{3} \left( T(h-1) + S^{\m}(h-1) \right) \right] \\
    &  & \mbox{} + \frac{1}{2}
\left[ \frac{2}{3}\, T(h-1) + \frac{1}{3} \left( T(h-1) + S^{\m}(h-1)
\right) \right] \enspace.
\end{eqnarray*}

\textbf{Equation (\ref{eqn-T}).}
Since $\evaluate(v)$ starts
with two calls to itself to compute~$x_1,x_2$, we get the first
term~$2\, T(h-2)$ on the right hand side. 
For the remaining terms, we consider two possible cases,
depending on whether the three children~$y_1, y_2, y_3$ of~$v$ are
equal. If~$y_1=y_2=y_3$, assuming that the children of~$y_1$ are not
all equal, and the same for the children of~$y_2$, we have
$$
\begin{array}{|c|c|c|}
\hline\multicolumn{3}{|c|}{T(h) \text{ given } y_1= y_2= y_3}\\\hline
\text{event}&\text{probability}&\text{complexity}\\
\hline
y_1=x_1, y_2=x_2&4/9&2S^{\rM}(h-1)\\
\hline
y_1=x_1, y_2\neq x_2&2/9&T(h-1)+S^{\rM}(h-1)\\
\hline
y_1\neq x_1, y_2=x_2&2/9&T(h-1)+S^{\rM}(h-1)\\
\hline
y_1\neq x_1, y_2\neq x_2&1/9&T(h-1)+S^{\m}(h-1)\\
\hline
\end{array}
$$
As before, the complexities are in non-decreasing order, and we
observe that Inequalities~(\ref{eqn-SMbd}) for height $h-1$ imply that
in a worst case input the children of~$y_1$ are not all equal, and that the
same holds for the children of~$y_2$.

If~$y_1, y_2, y_3$ are not all equal, we have three events
$y_1=y_2\neq y_3$, $y_1\neq y_2=y_3$ and $y_3=y_1\neq y_2$ each of
which occurs with probability $1/3$. This leads to the following
analyses
$$
\begin{array}{|c|c|c|}
\hline\multicolumn{3}{|c|}{T(h) \text{ given } y_1= y_2\neq y_3}\\\hline
\text{event}&\text{probability}&\text{complexity}\\
\hline
y_1=x_1, y_2=x_2&4/9&2S^{\rM}(h-1)\\
\hline
y_1=x_1, y_2\neq x_2&2/9&T(h-1)+S^{\rM}(h-1)+S^{\m}(h-1)\\
\hline
y_1\neq x_1, y_2=x_2&2/9&T(h-1)+S^{\rM}(h-1)+S^{\m}(h-1)\\
\hline
y_1\neq x_1, y_2\neq x_2&1/9&T(h-1)+2S^{\m}(h-1)\\
\hline
\end{array}
$$
$$
\begin{array}{|c|c|c|}
\hline\multicolumn{3}{|c|}{T(h) \text{ given } y_1\neq y_2= y_3}\\\hline
\text{event}&\text{probability}&\text{complexity}\\
\hline
y_1=x_1, y_2=x_2&4/9&T(h-1)+S^{\rM}(h-1)\\
\hline
y_1=x_1, y_2\neq x_2&2/9&T(h-1)+S^{\rM}(h-1)+S^{\m}(h-1)\\
\hline
y_1\neq x_1, y_2=x_2&2/9&T(h-1)+S^{\rM}(h-1)+S^{\m}(h-1)\\
\hline
y_1\neq x_1, y_2\neq x_2&1/9&T(h-1)+2S^{\m}(h-1)\\
\hline
\end{array}
$$
$$
\begin{array}{|c|c|c|}
\hline\multicolumn{3}{|c|}{T(h) \text{ given } y_3=y_1\neq y_2}\\\hline
\text{event}&\text{probability}&\text{complexity}\\
\hline
y_1=x_1, y_2=x_2&4/9&T(h-1)+S^{\rM}(h-1)\\
\hline
y_1=x_1, y_2\neq x_2&2/9&T(h-1)+S^{\rM}(h-1) + S^{\m}(h-1)\\
\hline
y_1\neq x_1, y_2=x_2&2/9&T(h-1)+S^{\m}(h-1)\\
\hline
y_1\neq x_1, y_2\neq x_2&1/9&T(h-1)+2S^{\m}(h-1)\\
\hline
\end{array}
$$
In all three events, we observe that Inequalities~(\ref{eqn-SMbd}) for
height $h-1$ imply that in a worst case input, the children of~$y_1$
are not all equal, and the same holds for the children of~$y_2$.

Applying Inequalities~(\ref{eqn-SMbd}) for height $h-1$, it follows that
line by line the complexities in the last three tables are at least
the complexities in the table for the case~$y_1 = y_2 = y_3$.
Therefore the worst case also corresponds to an input in which~$y_1,
y_2, y_3$ are not all equal.  We conclude Eq.~(\ref{eqn-T}) as before,
by taking the expectation of the complexities in the last three
tables.

\begin{theorem}
\label{thm:algoonetwo}
$T(h), S^{\rM}(h)$, and~$S^{\m}(h)$ are all in~$\Order(\alpha^h)$,
where $\alpha\leq 2.64944$.
\end{theorem}
\begin{proof}
We make an ansatz~$T(h)\leq a\,\alpha^h$, $S^{\rM}(h) \leq b\,
\alpha^h$, and~$S^{\m}(h)\leq c\, \alpha^h$, and find
constants~$a,b,c, \alpha$ for which we may prove these inequalities by
induction.

The base cases tell us that
$2 \leq c \alpha$, $\frac{3}{2} \leq b \alpha$,  $1 \leq a$, and $\frac{8}{3} \leq a \alpha$.

Assuming we have constants that satisfy these conditions, and that the
inequalities hold for all appropriate~$l < h$, for some~$h \geq 2$, we
derive sufficient conditions for the inductive step to go through.

By the induction hypothesis, Lemma~\ref{lemma:algoonetwo}, and the
ansatz, we have
\begin{eqnarray*}
S^{\m}(h)  & \leq & 
    a\, \alpha^{h-2} + a\, \alpha^{h-1} + \frac{2b}{3}\, 
    \alpha^{h-1} + \frac{c}{3}\, \alpha^{h-1} \enspace, \\
S^{\rM}(h) & \leq & 
    a\, \alpha^{h-2} + \frac{2a}{3}\, \alpha^{h-1} 
    + \frac{b}{3}\, \alpha^{h-1} + \frac{c}{3}\, \alpha^{h-1} \enspace,
                    \quad \textrm{and} \\
T(h) & \leq & 
    2a\, \alpha^{h-2} + \frac{23a}{27}\, \alpha^{h-1}
    + \frac{26}{27}\, \alpha^{h-1} + \frac{18}{27}\, \alpha^{h-1}
\enspace.
\end{eqnarray*}
These would imply the required bounds on~$S^{\m}(h), S^{\rM}(h), T(h)$
if
\begin{equation}
\label{eqn-ind} 
a + \tfrac{3a+2b+c}{3} \alpha  \leq  c\, \alpha^2 \enspace, \qquad
a + \tfrac{2a+b+c}{3} \alpha  \leq  b\, \alpha^2 \enspace, \textrm{ and}\qquad
2a + \tfrac{23a+26b+18c}{27} \alpha \leq  a\, \alpha^2 \enspace.
\end{equation}
The choice~$\alpha=2.64944$, $a = 1.02$, $b=0.559576 \times a$, and $c=0.755791 \times a$ satisfies the base case as well as all the
Inequalities~(\ref{eqn-ind}), so the theorem holds by induction.
\end{proof}

\section{Concluding remarks}

In this article, we revisited a technique due to Jayram, Kumar, and
Sivakumar for proving a lower bound on the decision tree complexity 
of~$\TMAJ_h$, the recursive majority-of-three function of height~$h$.
We showed that it could be enhanced by obtaining better estimates 
on the probability~$p_h^\delta$
with which the absolute minority variable is queried under the hard
distribution. The new estimates are obtained by considering highly
symmetric encodings of height-$(h-k)$ inputs into height-$h$ inputs.
The analysis of the encodings quickly becomes intractable with
growing~$k$. However, by appealing to the underlying symmetry in the
function, the analysis can be
executed explicitly for~$k = 1,2$, and with the aid of a computer
for~$k \le 4$. This leaves us with several immediate questions about the
technique:
\begin{enumerate}

\item Is there is a more efficient algorithm for the analysis?

\item Is there a succinct, explicit analysis for larger values of~$k$?

\item What is the best lower bound we can obtain using this technique?

\item Does this technique give us any intuition into more efficient 
algorithms?

\end{enumerate}

We also present a new (more efficient) algorithm for~$\TMAJ_h$ based on 
the idea that a partial evaluation of a formula helps us form an opinion 
about the value of its subformulae.  We use the opinions at a certain 
stage of the algorithm to choose the next variable to query.
Additionally, we use a depth-$2$ recursive
algorithm that is optimized for computing the value of a partially 
evaluated formula. It is likely that algorithms with depth~$k$ recursion,
with~$k > 2$ give us further improvements in efficiency. However, their 
analysis seems to be beyond the scope of the techniques used in this work.

\bibliographystyle{alpha}
\bibliography{general}

\appendix

\section{Python program}\label{code}

The following Python program is also available at:\\
\indent\url{https://www.dropbox.com/s/wcrdoib5h918p2e/commented-majority.py}\bigskip

{\small
\begin{lstlisting}[language=Python,tabsize=2,breaklines=true]
def prep(dep):
	global depth,g,d
	# depth = depth of recursive NOT-3-MAJ circuits considered
	# g[i] = list of depth i stable configurations for 0 <= i <= depth
	# for a record r representing a configuration of depth i
	#	r[0] = # of evaluations giving 0
	#	r[1] = # of evaluations giving 1
	#	r[2] = # of queried sensitive bits in all evaluations with correct output
	#	r[3:6] = indices of the three children, sorted (where applicable)
	# d[i] = dictionary that tells the index of a stable configuration of depth i from the indices of its 3 depth i-1 children
	# prep sets the values of these global variables (will not be changed)
	depth=dep
	g=[[[0,1,depth%2],[1,1,0]]]
	# two stable configurations in g[0]: a variable set to 1 and a not queried variable
	d=[{}]
	# empty dictionary in d[0]
	for i in range(1,depth+1):
		# building g[i] and d[i] under the names gg and dd
		gl=g[-1]
		gg=[[0,1,((depth-i)%2)*2**i]]
		# the first record represents a fully queried subtree of depth i evaluating to 1
		dd={}
		for a in range(len(gl)):
			for b in range(max(a,1),len(gl)):
				for c in range(b,len(gl)):
					# enforcing a <= b <= c and no two fully evaluated siblings in a stable configuration
					dd[(a,b,c)]=len(gg)
					A,B,C=gl[a],gl[b],gl[c]
					gg.append([A[0]*B[1]*C[1]+A[1]*B[0]*C[1]+A[1]*B[1]*C[0],A[1]*B[0]*C[0]+A[0]*B[1]*C[0]+A[0]*B[0]*C[1],A[0]*B[1]*C[2]+A[0]*B[2]*C[1]+A[1]*B[0]*C[2]+A[1]*B[2]*C[0]+A[2]*B[0]*C[1]+A[2]*B[1]*C[0],a,b,c])
					# inserting the current record in gg and dd - formula is long but simple
		g.append(gg)
		d.append(dd)
		# inserting the final gg and dd as g[i] and d[i]

def check(a0,a1):
	global alpha,ot,opt
	# alpha is considered as the rational a0/a1 but the integers are stored
	# ot is a table dynamically built for the strategy C maximizing rho_alpha(C)
	# ot[a][0] = # of inputs in which the optimal strategy C applied after configuration g[depth][a] queries abs. minority
	# ot[a][1] = # total # of sensitive bits revealed for C applied after g[depth][a]
	alpha=[a0,a1]
	ot=[[0,0]]
	# g[depth][0] is impossible
	for a in range(1,len(g[-1])):
		opt=[0,g[-1][a][2]]
		# opt is the best current guess for ot[a]
		# here it is set to value corresponding to stopping at the configuration g[depth][a]
		# opt[0] = 0 as the absolute minority is not queried in any stable configuration
		adj(depth,a,[],0)
		# here adj is applied to the root of the current configuration g[depth][a]
		# it crawls through the entire tree and updates opt if querying a variable is better than the current optimum
		ot.append(opt)
		# opt is now the correct optimum - it is inserted in the list
	return ot[-1]
	# here max_C rho_alpha(C) = ot[-1][1]/(2**depth*N)-alpha*ot[-1]/N, where N is the total # of 0-hard inputs
	# also: alpha(C) = ot[-1][1]/(ot[-1][0]*2**depth) (unless ot[-1]=[0,0] and C queries no input bit)

def adj(i,a,s,t):
	global ot,opt
	# here we consider what happens if a vertex W in a stable configuration is evaluated to 0 or 1
	# g[i][a] = the configuration below W
	# s is a list of depth-i pairs of siblings to add to g[i][a] to arrive to the stable depth d configuration considered
	# first goal compute tt as follows
	#	tt[0] # of inputs in which the optimal strategy C applied after W is set to 0 (instable) queries abs. minority
	#	tt[1] total # of sensitive bits revealed in same situation
	#	tt[2] # of consistent inputs with W on the root to absolute minority path
	# t = (same as tt but for the parent of W)
	if s==[]:
		tt=[0,2**i,1,0,0]
		# no input considered gives 0 at the root
	else:
		s1,s2=g[i][s[0][0]],g[i][s[0][1]]
		# the siblings of W
		nn,ne=s1[1]*s2[1],s1[0]*s2[1]+s2[0]*s1[1]
		tt=[nn*t[0]+ne*t[3],nn*t[1]+ne*t[4],nn*t[2]]
		# based on the rule to evaluate the siblings of W if W evaluates to 0
		if s[0][0]==0:
			tt.extend([s2[0]*t[0],s2[0]*t[1]])
			# based on the fact that if W evaluates to 1 and a so does a sibling, then the parent evaluates to 0
		else:
			# we have a stable configurationa and use d to find its index in g[depth]
			w,j=0,i
			for si in s:
				j=j+1
				if w<=si[0]:
					triple=(w,si[0],si[1])
				elif w<=si[1]:
					triple=(si[0],w,si[1])
				else:
					triple=(si[0],si[1],w)
				# here we sorted the three siblings w, si[0] and si[1]
				w=d[j][triple]
			# by now w is index of the full stable configuration and ot[w] contains the numbers we seek
			tt.extend(ot[w])
	if i>0:
		# if W is not a variable we recursively call adj on its non-evaluated children
		A=g[i][a]
		if A[3]>0:
			adj(i-1,A[3],[A[4:6]]+s,tt)
		adj(i-1,A[4],[[A[3],A[5]]]+s,tt)
		adj(i-1,A[5],[A[3:5]]+s,tt)
	else:
		# if W is a variable we compute the effects of querying it, compare to opt and update opt if needed
		pair=[tt[0]+tt[2]+tt[3],tt[1]+tt[4]]
		if (pair[0]-opt[0])*alpha[0]*2**depth<(pair[1]-opt[1])*alpha[1]:
			opt=pair

import fractions
def alpha(dep):
	# finds minimal alpha with rho_alpha(C) <= 0 for all C = max alpha(C) for non-empty C
	# by recursively applying check(alpha) that either gives a higher alpha or confirms that alpha is optimal
	prep(dep)
	p=0
	q=1
	x=check(p,q)
	while (x != [0,0]):
		f=fractions.Fraction(x[1],x[0]*(2**dep))
		p=f.numerator
		q=f.denominator
		x=check(p,q)
	print("For depth ",dep,", optimal alpha is ",f)
	print("leading to a lower bound of (",2+(p/q)**(-1/dep),")^h")
	
\end{lstlisting}
}

\end{document}